%% file: PEPO.tex
\documentclass{article}
\usepackage{geometry}
\usepackage[nospace,noadjust]{cite}
\usepackage{amsmath,amssymb,amsfonts,times}
\usepackage{graphicx}
\usepackage{textcomp}
\usepackage{xcolor}
\usepackage{subfigure}
\usepackage{epsfig}
\usepackage{multirow}
\usepackage{algorithm}
\usepackage{algorithmicx}
\usepackage{algpseudocode}
\usepackage{array}
\usepackage{epstopdf}
\usepackage{color}
\usepackage{diagbox}
\usepackage{bm}
\usepackage{bbm}
\usepackage{mathrsfs}
\usepackage{tcolorbox}
\usepackage{pdfsync}

\input{macro}

\usepackage{mathtools}

\newcommand\keywords[1]{\textbf{Keywords}: #1}

\title{\LARGE \bf Quantum State Tomography for\\ Tensor Networks in Two Dimensions}

\author{Zhen Qin  and Zhihui Zhu\thanks{ZQ (email: zhenqin@umich.edu) is with the with the Michigan Institute for Computational Discovery and Engineering, Department of Electrical Engineering and Computer Science, Department of Statistic, University of Michigan, Ann Arbor, MI 48109 USA, and also with the Department of Computer Science and Engineering, the Ohio State University, Columbus OH 43210 USA. ZZ (email: zhu.3440@osu.edu) is with the Department of Computer Science and Engineering, the Ohio State University, Columbus OH 43210 USA.}
}

\begin{document}

\maketitle

\begin{abstract}
Estimating quantum states from experimental data, typically conducted through quantum state tomography (QST), plays a crucial role in verifying the correctness and evaluating the performance of quantum devices. Nonetheless, executing QST for generic unstructured or low-rank quantum states demands an overwhelming number of state copies, escalating exponentially with the quantity of individual quanta within the system, even with the most optimized measurement configurations. %Recent insights indicate that quantum states in one dimension can be effectively approximated by matrix product operators (MPOs) possessing a matrix/bond dimension and established stringent limits on the requisite number of state copies for reconstructing MPO states, scaling polynomially with the qubit count.
Recent work has shown that for one-dimensional quantum states that can be effectively approximated by matrix product operators (MPOs), a polynomial number of copies of the state suffices for reconstruction.
Compared to MPOs in one dimension, projected entangled-pair states (PEPSs) and projected entangled-pair operators (PEPOs), which represent typical low-dimensional structures in two dimensions, are more prevalent as a looped tensor network. However, a formal analysis of the sample complexity required for estimating PEPS or PEPO has yet to be established. In this paper, we aim to address this gap by providing theoretical guarantees for the stable recovery of PEPS and PEPO. Our analysis primarily focuses on two quantum measurement schemes: $(i)$ informationally complete  positive operator valued measures (IC-POVMs), specifically the spherical $t$-designs ($t \geq 3$), and $(ii)$ projective rank-one measurements, in particular Haar random projective measurements. We first establish stable embeddings for PEPSs (or PEPOs) to ensure that the information contained in the states can be preserved under these two measurement schemes. We then show that a constrained least-squares estimator achieves stable recovery for PEPSs (or PEPOs), with the recovery error bounded when the number of state copies scales linearly under spherical $t$-designs and polynomially under Haar-random projective measurements with respect to the number of qudits. These results provide theoretical support for the reliable use of PEPS and PEPO in practical quantum information processing.
\end{abstract}

\keywords
Quantum state tomography (QST), projected entangled-pair states (PEPSs), projected entangled-pair operators (PEPOs), stable embedding, stable recovery, Haar random projective measurements, spherical $t$-designs.

\section{Introduction}
\label{Introduction}

Quantum state tomography (QST) \cite{bertrand1987tomographic,vogel1989determination,leonhardt1995quantum,hradil1997quantum}, widely regarded as the gold standard for benchmarking and verifying quantum devices, seeks to accurately determine the density matrix characterizing a quantum state. In systems comprising $n$ qudits, (which are $d$-level quantum systems; qubits have $d=2$), this state is denoted by a density matrix $\vrho$ of dimensions $d^\nqbit \times d^\nqbit$. Reconstructing the quantum state requires performing quantum measurements on a large number of identical copies of the state. Such physical measurements can be described through a Positive Operator-Valued Measure (POVM), composed of positive semi-definite (PSD) matrices or operators $\{\mA_1,\ldots,\mA_K\}$ that collectively sum to the identity operator. Each operator $\mA_k$ ($k=1,\dots, K$) in the POVM corresponds to a potential measurement outcome, with the probability of obtaining that outcome given by $p_k = \trace(\mA_k \vrho)$. Due to the probabilistic nature of quantum measurements, multiple measurements (say $M$) with the same POVM are necessary to obtain statistically accurate estimates $\wh p_k$ of each $p_k$. Without considering statistical errors, $\{p_k\}$ can be viewed as $K$ linear measurements of the state~$\vrho$. From the perspective of machine learning, we can categorize $\{p_k\}$ and their empirical estimates~${\wh p_k}$ as population and empirical measurements of the state, respectively.

Numerous algorithmic \cite{hradil1997quantum,vrehavcek2001iterative,blume2010optimal,granade2016practical,
lukens2020practical,blume2012robust,faist2016practical,kyrillidis2018provable,brandao2020fast,torlai2018neural,carleo2019machine,lohani2020machine,qin2025enhancing} and theoretical \cite{KuengACHA17,guctua2020fast,francca2021fast,liu2011universal, haah2017sample,voroninski2013quantum} research endeavors concerning QST have been extensively explored. It has been demonstrated that when measurements are conducted on one state at a time, estimating a general density matrix with an accuracy of $\delta$ in the trace norm between the reconstructed density matrix and the true density matrix requires a number of total state copies proportional to $O(d^{3n}/\delta^2)$ using independent measurements \cite{haah2017sample}. This rigorously illustrates the exponential growth in the requirement for state copies in QST as the number of qudits increases. To mitigate this demand,
we can leverage the inherent structure of pure or nearly pure quantum states characterized by low entropy and represented as low-rank density matrices. When utilizing low-rankness, \cite{haah2017sample,francca2021fast} demonstrated that previous requirements can be respectively reduced to $O(d^nr^2/\delta^2)$ for independent measurements.

In the realm of quantum computing, the dimensions of quantum computers have seen a rapid surge in recent years, with some of the most advanced processors now boasting over 100 qubits \cite{preskill2018quantum,arute2019quantum,chow2021ibm}, consequently, even for a rank-one density matrix--corresponding to a pure quantum state attainable solely by a noiseless quantum device--the number of state copies necessary for QST still scales exponentially with regard to $n$ qubits. Fortunately, matrix product operators (MPOs) have been identified as effective tools for representing states in noisy quantum systems, where noise inherently limits the degree of quantum entanglement, thereby enabling an efficient representation of quantum states \cite{noh2020efficient}. In addition, these states are also found in numerous quantum systems with short-range interactions, as well as states generated by such systems within a finite time frame \cite{eisert2010}. Recently, \cite{qin2024quantum,qin2024sample} successfully analyzed the polynomial scaling of the total number of state copies required for estimating an MPO by leveraging Haar random projective measurements and spherical $t$-designs to estimate the MPO. Specifically, for a ground-truth MPO, it was shown that using $O(n^3)$ or $O(n)$ state copies, respectively, with these two measurement strategies ensures, with high probability, that a properly constrained least-squares minimization based on the empirical measurements stably recovers the ground-truth state to within $\epsilon$-closeness in the Frobenius norm.

MPOs, while widely used to explore one-dimensional quantum systems, face significant challenges when extended to higher dimensions, such as in studying the complexities of interacting spin systems in two dimensions. One approach is to adapt MPOs by arranging spins linearly on a two-dimensional lattice. While this method has shown promise in some cases \cite{yan2011spin}, it struggles to accurately capture the intricate entanglement features found in typical two-dimensional ground states, particularly as the system size grows. Alternatively, one may employ two-dimensional tensor networks, such as projected entangled-pair states (PEPSs) and projected entangled-pair operators (PEPOs), which provide a complementary perspective by representing pairs of maximally entangled auxiliary systems confined to low-dimensional subspaces. PEPSs and PEPOs have been demonstrated to accurately capture a wide range of physical states, such as the thermal state \cite{alhambra2021locally}, 2D cluster state \cite{raussendorf2001one}, Toric Code model \cite{kitaev2003fault}, 2D resonating valence bond state \cite{anderson1987resonating}, and 2D AKLT model \cite{affleck1988valence,wei2011affleck}. A detailed exposition of PEPS and PEPO structures in two dimensions is provided in \cite{cirac2021matrix}.

To the best of our knowledge, there has been no theoretical analysis conducted regarding the determination of the number of state copies in QST for PEPSs and PEPOs, which remains an open question as highlighted in \cite{anshu2024survey}. While PEPS-based quantum state tomography has been implemented in \cite{guo2024quantum}, the absence of a canonical form in PEPSs and PEPOs, unlike MPOs with open boundary conditions, poses a significant challenge in theoretical analysis. This arises from the impossibility of selecting an orthonormal basis simultaneously for all bond indices \cite{orus2019tensor}. This limitation extends beyond PEPSs/PEPOs to encompass MPOs with periodic boundary conditions or whenever a loop is present in the tensor network. Conceptually, a loop in the tensor network denotes an inability to formally partition the network into distinct left and right segments by just cutting one index, rendering a Schmidt decomposition between left and right segments nonsensical. Consequently, the computational efficiency of PEPSs and PEPOs is hindered by these inherent challenges \cite{cramer2010efficient}.

\paragraph{Our contribution:} {In this paper, we aim to develop theoretical analysis on the sampling complexity for stable recovery of PEPS and PEPO. Specifically, we consider a PEPS composed of $n = qp$ qudits on a two-dimensional $q \times p$ lattice with $n = qp$
,  with a bond dimension at most $T$ and bounded factors. Similarly, a PEPO consists of $n = qp$ qudits on a two-dimensional $q \times p$ lattice with $n = qp$
, each having bond dimension at most $R$ and bounded factors. The degrees of freedom scales as $O(ndT^4)$ for PEPS and $O(nd^2R^4)$ for PEPO; a detailed representation is provided in \Cref{sec: defi of tensor networks 2D}.}

Our analysis focuses on two types of quantum measurements: (i) informationally complete positive operator-valued measures (IC-POVMs), realized via spherical $t$-designs ($t \geq 3$) \cite{scott2006tight}, and (ii) projective rank-one measurements, specifically Haar random projective measurements. An IC-POVM is defined such that its measurement statistics uniquely determine any quantum state and are commonly used to derive optimal sample complexities \cite{haah2017sample}. However, IC-POVMs are generally difficult to implement in practice, motivating the use of Haar random measurements. In many scenarios, a single projective rank-one POVM is insufficient to recover a general quantum state $\vrho$, even with an infinite number of measurement repetitions, as the resulting probabilities $\{ p_k \}$ only reveal the diagonal elements of $\vrho$ in the standard computational basis. To address this limitation, we employ $Q$ Haar random projective measurements and establish stable embedding results that can ensure exact recovery with $Q = \Omega(n d T^4 \log (1 + qp))$ for PEPSs and $Q = \Omega(n d^2 R^4 \log (1 + qp))$ for PEPOs, assuming zero statistical error in the measurements.

Next, we investigate the recovery of PEPS and PEPO from quantum measurements, accounting for statistical errors and establish recovery bounds concerning the number of state copies using both spherical $t$-designs and Haar random projective measurements. Specifically, we establish theoretical accuracy bounds for a constrained least-squares estimator used to recover PEPS and PEPO. The main results are presented in the following two theorems.

\begin{theorem}[informal version of \Cref{final conclusion of 3-designs theorem PEPS,final conclusion of Haar theorem PEPS}]
\label{Statistical Error_Haar_Measurement_informal PEPS}
Given an $n$-qudit PEPS state on a two-dimensional $q \times p$ lattice with $n = qp$, bond dimension $T$, and bounded factors,
we either $(i)$ generate a set of proper $t$-designs ($t\geq 3$) and measure the state $M$ times, or $(ii)$ randomly generate $Q$ Haar random projective measurement bases and measure the state in each basis $M$ times. For any $\delta>0$, assume the number of total state copies satisfies
\begin{eqnarray}
    \label{upper bound of_QM for all cases in PEPS intro}
    M&\!\!\!\!=\!\!\!\!& \Omega(n d T^4 \log (1 + qp)/\delta^2),\ \text{$t$-designs}, \\
    QM &\!\!\!\!=\!\!\!\!& \Omega(n^3 d T^4 \log (1 + qp)/\delta^2) , Q= \Omega(n d T^4 \log (1 + qp)), \  \textup{Haar}.
\end{eqnarray}
Then, with high probability, a properly constrained least-squares minimization with the quantum measurements stably recovers the ground-truth state with $\delta$-closeness in the trace norm.
\label{thm:informal-main PEPS}\end{theorem}

\begin{theorem}[informal version of \Cref{final conclusion of Haar theorem,final conclusion of recovery error1 for t-design larger 3}]
\label{Statistical Error_Haar_Measurement_informal PEPO}
Given an $n$-qudit PEPO state on a two-dimensional $q \times p$ lattice with $n = qp$, bond dimension $R$, and bounded factors, we either $(i)$ generate a set of proper $t$-designs ($t\geq 3$) and measure the state $M$ times, or  $(ii)$ randomly generate $Q$ Haar random projective measurement bases and measure the state in each basis $M$ times. For any $\epsilon>0$, assume the number of total state copies satisfies
\begin{eqnarray}
    \label{upper bound of_QM for all cases in PEPO intro}
    M&\!\!\!\!=\!\!\!\!& \Omega(n d^2 R^4 \log (1 + qp)/\epsilon^2),\ \text{$t$-designs}, \\
    QM &\!\!\!\!=\!\!\!\!& \Omega(n^3 d^2 R^4 \log (1 + qp)/\epsilon^2) , Q= \Omega(n d^2 R^4 \log (1 + qp)), \  \textup{Haar}.
\end{eqnarray}
Then, with high probability, a properly constrained least-squares minimization with the quantum measurements stably recovers the ground-truth state with $\epsilon$-closeness in the Frobenius norm.
\label{thm:informal-main PEPO}\end{theorem}
Our results ensure stable recovery with the total number of state copies growing only polynomially in the number of qudits $n$ using spherical $t$-designs or Haar measures. Similar to the MPO structure, employing the PEPS and PEPO structures significantly reduces the required number of state copies from $d^n$ to $\text{poly}(n)$.  Moreover, no restriction is placed on the number of state copies ($M$) associated with each Haar measurement basis. Consequently, our results provide theoretical justification for the practical use of single-shot measurements ($M=1$), a strategy already adopted in prior studies \cite{wang2020scalable,huang2020predicting}. Our recovery guarantees rely on stable embedding results and are established in the trace norm for PEPS and the Frobenius norm for PEPO. While our guarantees are initially established in the Frobenius norm for PEPO, they can be transferred to the trace norm in the low-rank setting by exploiting the robust relation between trace distance and Hilbert-Schmidt distance \cite{coles2019strong}.

\subsection{Notation}
\label{sec:notation}

We use bold capital letters (e.g., $\mA$) to denote matrices, except that $\mX$ denotes a tensor,  and bold lowercase letters (e.g., $\va$) to denote column vectors, and italic letters (e.g., $a$) to denote scalar quantities. Elements of matrices and tensors are denoted in parentheses. For example, $\mX(i_1, i_2, i_3)$ denotes the element in position
$(i_1, i_2, i_3)$ of the order-3 tensor $\mX$. The calligraphic letter $\calA$ is reserved for the linear measurement map.  For a positive integer $K$, $[K]$ denotes the set $\{1,\dots, K \}$. The superscripts $(\cdot)^\top$ and $(\cdot)^\dagger$ denote the transpose and Hermitian transpose, respectively \footnote{As is conventional in the quantum physics literature (but not in information theory and signal processing), we use $(\cdot)^\dagger$ to denote the Hermitian transpose.}. For two matrices $\mA,\mB$ of the same size, $\innerprod{\mA}{\mB} = \trace(\mA^\dagger\mB)$ denotes the inner product between them.
$\|\mA\|_F$, $\|\mA\|_1$ and $\|\mA\|_{\infty}$ represent the Frobenius norm, trace norm and maximum norm of $\mA$.
For two positive quantities $a,b\in \real$, the inequality $b\lesssim a$ or $b = O(a)$ means $b\leq c a$ for some universal constant $c$; likewise, $b\gtrsim a$ or $b = \Omega(a)$ represents $b\ge ca$ for some universal constant $c$.

\section{Preliminaries of Quantum State Tomography}
\label{sec:QST}

\subsection{Quantum measurements}
\label{sec: quantum measurements }

Quantum information science employs quantum states as tools for information processing \cite{nielsen2001quantum}. In quantum systems, the quantum state is typically described using a density operator. Specifically, when dealing with a quantum system comprising $n$ qudits, the density operator is expressed as $\vrho\in\C^{d^n\times d^n}$.
In all cases, the density matrix $\vrho$ must satisfy two essential conditions:  $(i)$~ $\vrho \succeq \vzero$ is a positive semidefinite (PSD) matrix, and $(ii)$~$\trace(\vrho) = 1$.

Quantum state tomography (QST) aims to construct or estimate the density matrix $\vrho$ by performing measurements on an ensemble of identical quantum states. Given the inherently probabilistic nature of quantum measurements, which are described using Positive Operator Valued Measures (POVMs) \cite{nielsen2002quantum}, it becomes necessary to prepare multiple copies of the quantum state in experiments. Specifically, a POVM consists of a set of positive semi-definite (PSD) matrices:
\begin{eqnarray}
\label{The defi of POVM quibt}
\{\mA_1,\ldots,\mA_K \}\in\C^{d^n\times d^n}, \ \ \st  \sum_{k=1}^K \mA_k = \mId_{d^n}.
\end{eqnarray}
Each POVM element $\mA_k$ is associated with a potential outcome of a quantum measurement. The probability $p_k$ of detecting the $k$-th outcome while measuring the density matrix $\rho$ is expressed as:
\begin{eqnarray}
\label{The defi of POVM 2 quibt}
p_k = \innerprod{\mA_k}{\vrho},
\end{eqnarray}
where $\sum_{k=1}^K p_k=1$ due to \eqref{The defi of POVM quibt} and $\trace(\vrho) = 1$. Due to the impracticality of directly collecting $\{p_k\}$ in a physical experiment, the measurement process is typically repeated $M$ times. Through averaging the statistically independent outcomes, empirical probabilities are obtained:
\begin{equation}
\wh p_{k} = \frac{f_k}{M}, \  k \in[K]:=\{1,\ldots,K\},
\label{eq:empirical-prob  quibt}
\end{equation}
where $f_k$ denotes the number of times the $k$-th outcome is observed in the $M$ experiments. Notably, QST relies on the use of empirical probabilities $\{\wh p_k\}$ to recover or estimate the unknown density operator $\vrho$. The variables $f_1,\ldots,f_K$ collectively follow a multinomial distribution $\operatorname{Multinomial}(M, \vp)$ \cite{severini2005elements}, where $M$ is the total number of measurements and $\vp = \begin{bmatrix} p_1 & \cdots & p_K\end{bmatrix}^\top$, where $p_k$ is defined in \eqref{The defi of POVM 2 quibt}. Consequently, the empirical probability $\wh p_k$ in \eqref{eq:empirical-prob quibt} acts as an unbiased estimator for the probability $p_k$. In the realms of information theory and signal processing, $\{p_k\}$ and $\{\wh p_k\}$ are commonly denoted as the population and empirical (linear) measurements, respectively.

As discussed previously, a single POVM may be insufficient to recover a general quantum state $\vrho$. To address this, consider $Q$ POVMs indexed by $q\in[Q]$, where each POVM is represented as $\{A_{q,k}\}_{k\in[K]}$ consisting of $K$ PSD operators. Each POVM is probed with $M$ shots, yielding the empirical frequency vectors $\wh\vp_q= [\wh p_{q,1} \cdots \wh p_{q,K}]^\top$ for $q \in [Q]$. To streamline the notation, the probabilities associated with each POVM, $\{ \<A_{q,k}, \rho\>\}$ can be represented through a linear map $\calA_q:  \C^{d^n\times  d^n} \rightarrow \R^K$ defined as
\begin{eqnarray}
\label{The defi of POVM element in q-th POVM}
 \calA_q(\vrho) =  \begin{bmatrix}
          \< \mA_{q,1}, \vrho  \> \\
          \vdots \\
          \< \mA_{q,K}, \vrho  \>
        \end{bmatrix}.
\end{eqnarray}
By stacking the linear operators $\{\calA_q\}$ corresponding to the $Q$ POVMs into a single linear map $ \calA:\C^{d^n\times d^n} \rightarrow \R^{KQ}$, we obtain $KQ$ population measurements expressed as
\begin{eqnarray}
\label{The defi of population measurement in Q cases (K measurements)}
 \vp = \calA(\vrho)  = \begin{bmatrix}
          {\bm p}_{1} \\
          \vdots \\
          {\bm p}_{Q}
        \end{bmatrix}  = \begin{bmatrix}
          \calA_1(\vrho) \\
          \vdots \\
          \calA_Q(\vrho)
        \end{bmatrix}.
\end{eqnarray}

For each POVM, we repeat the measurement process $M$ times and stack  all the total empirical measurements together as
\begin{equation}
\wh\vp = \begin{bmatrix}
    \wh \vp_1 \\ \vdots \\ \wh \vp_M
\end{bmatrix}.
\label{eq:map-M-POVM1}\end{equation}
Note that for a spherical $t$-design POVM, $Q = 1$ and $K \geq d^{2n}$, whereas for Haar random projective measurements, we have $Q \geq 1$ and $K = d^n$.

\subsection{Tensor Networks in Two Dimensions}
\label{sec: defi of tensor networks 2D}

\paragraph{Projected Entangled-Pair State (PEPS)}
For a pure state $\vrho = \vu \vu^\dagger \in\C^{d^\nqbit\times d^ \nqbit}$  in an $\nqbit$-qudit quantum system arranged on a two-dimensional $q \times p$ lattice with $n = qp$ and $\vu\in\C^{d^\nqbit}$, we use a single index-array $i_{\ul{1}\, \ul{1}}\cdots i_{\ul{q}\, \ul{p}}$ to specify the indices, where $i_{\ul{1}\, \ul{1}},\ldots,i_{\ul{q}\, \ul{p}}\in[d]$.\footnote{Specifically, $i_{\ul{1}\, \ul{1}}\cdots i_{\ul{q}\, \ul{p}}$ represents $(i_{\ul{1}\, \ul{1}} + \sum_{a = 2}^p d^{a-1} (i_{\ul{1}\, \ul{a}} - 1) + \sum_{b=2}^q\sum_{c = 1}^p d^{(b-1)p+c-1}(i_{\ul{b}\, \ul{c}} - 1))$.
} We say that a quantum state $\vrho$ a Projected Entangled-Pair State (PEPS) if the $i_{\ul{1}\, \ul{1}}\cdots i_{\ul{q}\, \ul{p}}$-th component of $\vu$  can be expressed as follows~\cite{verstraete2008matrix,cirac2009renormalization,schuch2013condensed,orus2014practical,kshetrimayum2017simple}:

\begin{equation}
\label{DefOfPEPS}
\begin{split}
&\vu\big(i_{\ul{1}\,\ul{1}}\cdots i_{\ul{q}\,\ul{p}}\big)\\ &= \sum_{\substack{
      s_{\ul{a}\,\ul{b-1},\,\ul{a}\,\ul{b}}\;
      s_{\ul{a-1}\,\ul{b},\,\ul{a}\,\ul{b}} \\
      s_{\ul{a}\,\ul{b},\,\ul{a}\,\ul{b+1}}\;
      s_{\ul{a}\,\ul{b},\,\ul{a+1}\,\ul{b}} \\
      a\in[q],\, b\in[p]}}
   \hspace{-0.5cm}
   \mU_{\ul{1}\,\ul{1}}^{\,i_{\ul{1}\,\ul{1}}}
     (s_{\ul{1}\,\ul{0},\,\ul{1}\,\ul{1}},\,
      s_{\ul{0}\,\ul{1},\,\ul{1}\,\ul{1}},\,
      s_{\ul{1}\,\ul{1},\,\ul{1}\,\ul{2}},\,
      s_{\ul{1}\,\ul{1},\,\ul{2}\,\ul{1}} )
      \cdots \mU_{\ul{q}\,\ul{p}}^{\,i_{\ul{q}\,\ul{p}} }
     (s_{\ul{q}\,\ul{p-1},\,\ul{q}\,\ul{p}},\,
      s_{\ul{q-1}\,\ul{p},\,\ul{q}\,\ul{p}},\,
      s_{\ul{q}\,\ul{p},\,\ul{q}\,\ul{p+1}},\,
      s_{\ul{q}\,\ul{p},\,\ul{q+1}\,\ul{p}}) \\
&= \sum_{\substack{
      s_{\ul{a}\,\ul{b-1},\,\ul{a}\,\ul{b}}\;
      s_{\ul{a-1}\,\ul{b},\,\ul{a}\,\ul{b}} \\
      s_{\ul{a}\,\ul{b},\,\ul{a}\,\ul{b+1}}\;
      s_{\ul{a}\,\ul{b},\,\ul{a+1}\,\ul{b}} \\
      a\in[q],\, b\in[p]}}
   \;\Pi_{a=1}^{q} \Pi_{b=1}^{p}
   \mU_{\ul{a}\,\ul{b}}^{\,i_{\ul{a}\,\ul{b}} }
   (s_{\ul{a}\,\ul{b-1},\,\ul{a}\,\ul{b}},\,
    s_{\ul{a-1}\,\ul{b},\,\ul{a}\,\ul{b}},\,
    s_{\ul{a}\,\ul{b},\,\ul{a}\,\ul{b+1}},\,
    s_{\ul{a}\,\ul{b},\,\ul{a+1}\,\ul{b}}),
\end{split}
\end{equation}
where core factors $\mU_{\ul{a}\,\ul{b}}^{\,i_{\ul{a}\,\ul{b}}}$ are $4$-order complex tensors with size
$t_{\ul{a}\,\ul{b-1},\,\ul{a}\,\ul{b}}\times t_{\ul{a-1}\,\ul{b},\,\ul{a}\,\ul{b}}
\times t_{\ul{a}\,\ul{b},\,\ul{a}\,\ul{b+1}}\times t_{\ul{a}\,\ul{b},\,\ul{a+1}\,\ul{b}}$,
$a\in[q], b\in[p]$, and
$t_{\ul{a}\,\ul{0},\,\ul{a}\,\ul{1}} =
 t_{\ul{0}\,\ul{b},\,\ul{1}\,\ul{b}} =
 t_{\ul{a}\,\ul{p},\,\ul{a}\,\ul{p+1}} =
 t_{\ul{q}\,\ul{b},\,\ul{q+1}\,\ul{b}} = 1$.
To streamline notation, we denote the above format shortly as
$\vu = [\mU_{\ul{a}\,\ul{b}}]_{a=1,b=1}^{q,p}$. The bond dimensions of the PEPS in quantum physics are defined as follows:
\begin{eqnarray}
    \label{bond dimensions in PEPS}
    \mT = \begin{bmatrix}
    t_{\ul{1}\,\ul{1},\,\ul{1}\,\ul{2}} & t_{\ul{1}\,\ul{1},\,\ul{2}\,\ul{1}} & t_{\ul{2}\,\ul{1},\,\ul{2}\,\ul{2}} & \cdots & t_{\ul{q}\,\ul{1},\,\ul{q}\,\ul{2}} \\
    t_{\ul{1}\,\ul{2},\,\ul{1}\,\ul{3}} & t_{\ul{1}\,\ul{2},\,\ul{2}\,\ul{2}} & t_{\ul{2}\,\ul{2},\,\ul{2}\,\ul{3}} & \cdots & t_{\ul{q}\,\ul{2},\,\ul{q}\,\ul{3}} \\
    \vdots    & \vdots    & \vdots    & \ddots & \vdots    \\
    t_{\ul{1}\,\ul{p-1},\,\ul{1}\,\ul{p}} & t_{\ul{1}\,\ul{p-1},\,\ul{2}\,\ul{p-1}} & t_{\ul{2}\,\ul{p-1},\,\ul{2}\,\ul{p}} & \cdots & t_{\ul{q}\,\ul{p-1},\,\ul{q}\,\ul{p}} \\
    0 & t_{\ul{1}\,\ul{p},\,\ul{2}\,\ul{p}} & 0 & \cdots & 0
    \end{bmatrix}\in\R^{p\times (2q-1)}.
\end{eqnarray}
Note that the bond dimensions are equivalent to the ranks of the tensor network, so we denote $\text{rank}(\vu) = \mT$.

\begin{figure*}[t]
\centering
\includegraphics[width=13cm, keepaspectratio]%
{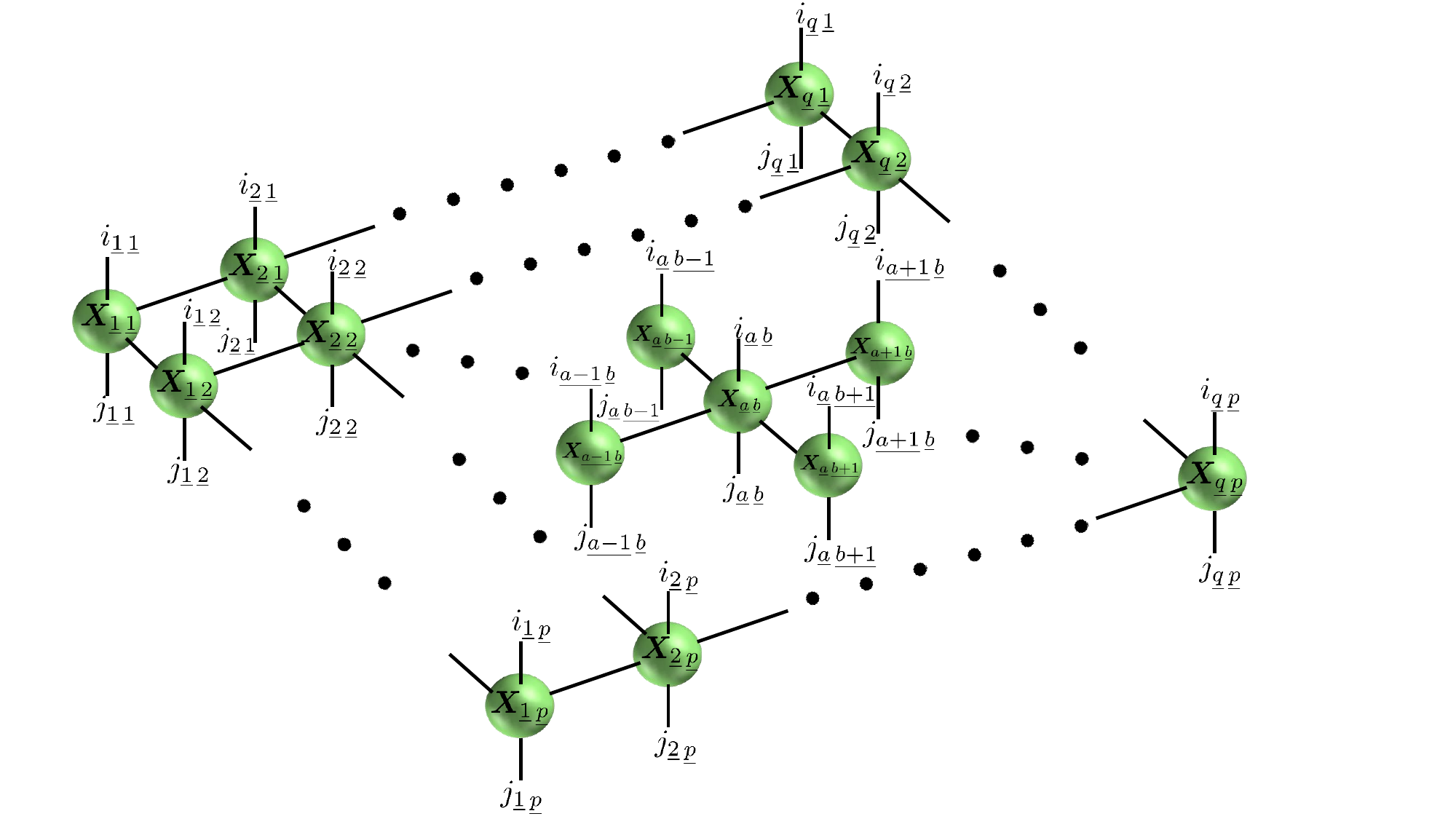}
\vspace{-0cm}
\caption{Illustration of the PEPO in \eqref{DefOfPEPO} from each element of the density matrix is presented in a diagrammatic form, where the line connecting two circles signifies the tensor contraction operation \cite{cichocki2014tensor}, and unconnected line segments denote indices.}
\label{The PEPO fig}
\end{figure*}

\paragraph{Projected Entangled-Pair Operator (PEPO)} This definition naturally extends to mixed states $\vrho \in \C^{d^\nqbit \times d^\nqbit}$, where we use $i_{\ul{1},\ul{1}} \cdots i_{\ul{q},\ul{p}}$ and $j_{\ul{1},\ul{1}} \cdots j_{\ul{q},\ul{p}}$ to denote the row and column indices, respectively. Specifically, we say that $\vrho$ is a Projected Entangled-Pair Operator (PEPO) on the $q \times p$ lattice if its $(i_{\ul{1},\ul{1}} \cdots i_{\ul{q},\ul{p}},, j_{\ul{1},\ul{1}} \cdots j_{\ul{q},\ul{p}})$-th entry can be expressed as follows:
%Similar to the PEPS case, for a density matrix $\vrho\in\C^{d^\nqbit\times d^ \nqbit}$ in an $\nqbit$-qudit quantum system, with $n = qp$ composed of positive integers $q$ and $p$, we use a single index-array $i_{\ul{1}\, \ul{1}}\cdots i_{\ul{q}\, \ul{p}}$ ($j_{\ul{1}\, \ul{1}}\cdots j_{\ul{q}\, \ul{p}}$) to specify the indices of rows (columns), where $i_{\ul{1}\, \ul{1}},\ldots,i_{\ul{q}\, \ul{p}}\in[d]$. Then, we say that $\vrho$ is a PEPO in the $q\times p$ lattice if its $(i_{\ul{1}\, \ul{1}}\cdots i_{\ul{q}\, \ul{p}},\ j_{\ul{1}\, \ul{1}}\cdots j_{\ul{q}\, \ul{p}})$-th element can be expressed as follows:
\begin{equation}
\label{DefOfPEPO}
\begin{split}
&\vrho\big(i_{\ul{1}\,\ul{1}}\cdots i_{\ul{q}\,\ul{p}}, \; j_{\ul{1}\,\ul{1}}\cdots j_{\ul{q}\,\ul{p}}\big)\\
&= \sum_{\substack{
      s_{\ul{a}\,\ul{b-1},\,\ul{a}\,\ul{b}}\;
      s_{\ul{a-1}\,\ul{b},\,\ul{a}\,\ul{b}} \\
      s_{\ul{a}\,\ul{b},\,\ul{a}\,\ul{b+1}}\;
      s_{\ul{a}\,\ul{b},\,\ul{a+1}\,\ul{b}} \\
      a\in[q],\, b\in[p]}}
   \hspace{-0.5cm}
   \mX_{\ul{1}\,\ul{1}}^{\,i_{\ul{1}\,\ul{1}},\,j_{\ul{1}\,\ul{1}}}
     (s_{\ul{1}\,\ul{0},\,\ul{1}\,\ul{1}},\,
      s_{\ul{0}\,\ul{1},\,\ul{1}\,\ul{1}},\,
      s_{\ul{1}\,\ul{1},\,\ul{1}\,\ul{2}},\,
      s_{\ul{1}\,\ul{1},\,\ul{2}\,\ul{1}} )
      \cdots \mX_{\ul{q}\,\ul{p}}^{\,i_{\ul{q}\,\ul{p}},\,j_{\ul{q}\,\ul{p}}}
     (s_{\ul{q}\,\ul{p-1},\,\ul{q}\,\ul{p}},\,
      s_{\ul{q-1}\,\ul{p},\,\ul{q}\,\ul{p}},\,
      s_{\ul{q}\,\ul{p},\,\ul{q}\,\ul{p+1}},\,
      s_{\ul{q}\,\ul{p},\,\ul{q+1}\,\ul{p}}) \\
&= \sum_{\substack{
      s_{\ul{a}\,\ul{b-1},\,\ul{a}\,\ul{b}}\;
      s_{\ul{a-1}\,\ul{b},\,\ul{a}\,\ul{b}} \\
      s_{\ul{a}\,\ul{b},\,\ul{a}\,\ul{b+1}}\;
      s_{\ul{a}\,\ul{b},\,\ul{a+1}\,\ul{b}} \\
      a\in[q],\, b\in[p]}}
   \;\Pi_{a=1}^{q} \Pi_{b=1}^{p}
   \mX_{\ul{a}\,\ul{b}}^{\,i_{\ul{a}\,\ul{b}}, j_{\ul{a}\,\ul{b}}}
   (s_{\ul{a}\,\ul{b-1},\,\ul{a}\,\ul{b}},\,
    s_{\ul{a-1}\,\ul{b},\,\ul{a}\,\ul{b}},\,
    s_{\ul{a}\,\ul{b},\,\ul{a}\,\ul{b+1}},\,
    s_{\ul{a}\,\ul{b},\,\ul{a+1}\,\ul{b}}).
\end{split}
\end{equation}
Here, the core tensors $\mX_{\ul{a}\,\ul{b}}^{\,i_{\ul{a}\,\ul{b}},\,j_{\ul{a}\,\ul{b}}}$ are four-way complex arrays with dimensions $r_{\ul{a}\,\ul{b-1},\,\ul{a}\,\ul{b}}\times r_{\ul{a-1}\,\ul{b},\,\ul{a}\,\ul{b}}
\times r_{\ul{a}\,\ul{b},\,\ul{a}\,\ul{b+1}}\times r_{\ul{a}\,\ul{b},\,\ul{a+1}\,\ul{b}}$,
subject to boundary conditions $r_{\ul{a}\,\ul{0},\,\ul{a}\,\ul{1}} =
 r_{\ul{0}\,\ul{b},\,\ul{1}\,\ul{b}} =
 r_{\ul{a}\,\ul{p},\,\ul{a}\,\ul{p+1}} =
 r_{\ul{q}\,\ul{b},\,\ul{q+1}\,\ul{b}} = 1$.
For brevity, we denote this PEPO representation as $\vrho = [\mX_{\ul{a}\,\ul{b}}]_{a=1,b=1}^{q,p}$.
An illustration of the PEPO structure is provided in Figure~\ref{The PEPO fig}. It is worth noting that when $j_{\ul{a}\,\ul{b}} = 1$ for $a\in[q]$ and $b\in[p]$, the PEPO reduces to the corresponding PEPS structure.
The bond dimensions correspond to the ranks along each tensor contraction direction, summarized by
\begin{eqnarray}
    \label{bond dimensions}
    \mR = \begin{bmatrix}
    r_{\ul{1}\,\ul{1},\,\ul{1}\,\ul{2}} & r_{\ul{1}\,\ul{1},\,\ul{2}\,\ul{1}} & r_{\ul{2}\,\ul{1},\,\ul{2}\,\ul{2}} & \cdots & r_{\ul{q}\,\ul{1},\,\ul{q}\,\ul{2}} \\
    r_{\ul{1}\,\ul{2},\,\ul{1}\,\ul{3}} & r_{\ul{1}\,\ul{2},\,\ul{2}\,\ul{2}} & r_{\ul{2}\,\ul{2},\,\ul{2}\,\ul{3}} & \cdots & r_{\ul{q}\,\ul{2},\,\ul{q}\,\ul{3}} \\
    \vdots    & \vdots    & \vdots    & \ddots & \vdots    \\
    r_{\ul{1}\,\ul{p-1},\,\ul{1}\,\ul{p}} & r_{\ul{1}\,\ul{p-1},\,\ul{2}\,\ul{p-1}} & r_{\ul{2}\,\ul{p-1},\,\ul{2}\,\ul{p}} & \cdots & r_{\ul{q}\,\ul{p-1},\,\ul{q}\,\ul{p}} \\
    0 & r_{\ul{1}\,\ul{p},\,\ul{2}\,\ul{p}} & 0 & \cdots & 0
    \end{bmatrix}\in\R^{p\times (2q-1)}.
\end{eqnarray}
Accordingly, we set $\text{rank}(\vrho) = \mR$ to indicate the tensor network ranks along all contraction directions.

\paragraph{Linear combination of PEPOs}
Next, we will define the linear combination of two PEPOs. In particular, for any two PEPOs
$\wt \vrho, \wh \vrho \in\C^{d^n \times d^n}$ of the form \eqref{DefOfPEPO} with core factors
$[\wt \mX_{\ul{a}\,\ul{b}}]_{a=1,b=1}^{q,p}$,
$[\wh \mX_{\ul{a}\,\ul{b}}]_{a=1,b=1}^{q,p}$
and $\text{rank}(\wt \vrho) = \text{rank}(\wh \vrho) = \mR$,
the $(i_{\ul{1}\,\ul{1}}\cdots i_{\ul{q}\,\ul{p}}, j_{\ul{1}\,\ul{1}}\cdots j_{\ul{q}\,\ul{p}})$-element
of their summation $\vrho = \wt \vrho + \wh \vrho$ can be expressed by the following core factors:
\begin{eqnarray}
    \label{middle tensor1}
    \mX_{\ul{a}\,\ul{b}}^{\,i_{\ul{a}\,\ul{b}},\,j_{\ul{a}\,\ul{b}}}
    \in\C^{2r_{\ul{a}\,\ul{b-1},\,\ul{a}\,\ul{b}}\times
          2r_{\ul{a-1}\,\ul{b},\,\ul{a}\,\ul{b}}\times
          2r_{\ul{a}\,\ul{b},\,\ul{a}\,\ul{b+1}}\times
          2r_{\ul{a}\,\ul{b},\,\ul{a+1}\,\ul{b}}},\
    a\in[q],b\in[p],
\end{eqnarray}
where the nonzero elements of core factors are respectively defined as:
\begin{eqnarray}
    \label{middle tensor2}
    \begin{cases}
    \mX_{\ul{a}\,\ul{b}}^{\,i_{\ul{a}\,\ul{b}},\,j_{\ul{a}\,\ul{b}}}
    (1:r_{\ul{a}\,\ul{b-1},\,\ul{a}\,\ul{b}},\ 1:r_{\ul{a-1}\,\ul{b},\,\ul{a}\,\ul{b}},\
     1:r_{\ul{a}\,\ul{b},\,\ul{a}\,\ul{b+1}},\ 1:r_{\ul{a}\,\ul{b},\,\ul{a+1}\,\ul{b}})
     = \wt\mX_{\ul{a}\,\ul{b}}^{\,i_{\ul{a}\,\ul{b}},\,j_{\ul{a}\,\ul{b}}},\\[0.3em]
    \mX_{\ul{a}\,\ul{b}}^{\,i_{\ul{a}\,\ul{b}},\,j_{\ul{a}\,\ul{b}}}
    (r_{\ul{a}\,\ul{b-1},\,\ul{a}\,\ul{b}}+1:2r_{\ul{a}\,\ul{b-1},\,\ul{a}\,\ul{b}},\
     r_{\ul{a-1}\,\ul{b},\,\ul{a}\,\ul{b}}+1:2r_{\ul{a-1}\,\ul{b},\,\ul{a}\,\ul{b}},\\
    \hspace{2.0cm} r_{\ul{a}\,\ul{b},\,\ul{a}\,\ul{b+1}}+1:2r_{\ul{a}\,\ul{b},\,\ul{a}\,\ul{b+1}},\
     r_{\ul{a}\,\ul{b},\,\ul{a+1}\,\ul{b}}+1:2r_{\ul{a}\,\ul{b},\,\ul{a+1}\,\ul{b}})
     = \wh\mX_{\ul{a}\,\ul{b}}^{\,i_{\ul{a}\,\ul{b}},\,j_{\ul{a}\,\ul{b}}}.
    \end{cases}
\end{eqnarray}
Thus, the linear combination of two PEPOs yields a new PEPO with ranks $\text{rank}(\vrho ) \leq \text{rank}(\wt\vrho ) + \text{rank}(\wh\vrho )$. Similarly, for PEPSs $\wh\vu = [\wh\mU_{\ul{a}\,\ul{b}}]_{a=1,b=1}^{q,p}$ and $\wt\vu = [\wt\mU_{\ul{a}\,\ul{b}}]_{a=1,b=1}^{q,p}$, it also holds that $\text{rank}(\wh\vu + \wt\vu ) \leq \text{rank}(\wh\vu ) + \text{rank}(\wt\vu ) $.

\subsection{Measurement settings}
\label{sec: stable embedding of PEPO}

\paragraph{Informationally Complete POVMs--spherical $t$-design POVM}

A central question in quantum state tomography concerns the choice of measurements that allow for the reliable and efficient recovery of an unknown state. Among various possibilities, \emph{informationally complete} POVMs (IC-POVMs) play a fundamental role, as their measurement statistics are sufficient to uniquely determine any density operator.
Specifically, if a POVM's statistics can uniquely determine any quantum state with a fixed dimension of $d^{2n}-1$, it is considered informationally complete \cite{prugovevcki1977information,busch1989determination,peres1997quantum} for an unknown density operator $\vrho\in\C^{d^n\times d^n}$ with $\trace(\vrho)=1$. So an informationally complete POVM  necessitates a minimum of $d^{2n}-1$ independent measurement outcomes. Furthermore, each element within the rank-one IC-POVM shares identical properties, thereby classifying it as one of the tight IC-POVMs \cite{scott2006tight}. In this part, our focus lies predominantly on (tight) spherical $t$-designs POVMs \cite{zauner1999quantum,renes2004symmetric,dall2014accessible}, which belong to tight rank-one IC-POVMs as defined in \cite[Definition 5 and Proposition 13]{scott2006tight}. To begin, let us provide the definition of spherical $t$-designs:
\begin{defi}
\label{definition_of_T_Design} (Spherical $t$-designs \cite{matthews2009distinguishability,dall2014accessible,KuengACHA17}). A finite set $\{\vw_k  \}_{k=1}^K\subset \C^{d^n}$ of normalized vectors is called a  spherical quantum $t$-design if \footnote{ $K\geq C_{d^n + \lfloor t/2 \rfloor - 1}^{\lfloor t/2 \rfloor} \cdot C_{d^n + \lceil t/2 \rceil - 1}^{\lceil t/2 \rceil} $ is  necessary to form a spherical $t$-design\cite{scott2006tight,gross2015partial}.}
\begin{eqnarray} \label{the definition of t designs}
     \frac{1}{K}\sum_{k=1}^K (\vw_k\vw_k^\dagger)^{\otimes s}  = \int (\vw\vw^\dagger)^{\otimes s} d\vw
\end{eqnarray}
holds for any $s\leq t$, where the integral on the right hand side is taken with respect to the Haar measure on the
complex unit sphere in $\C^{d^n}$.
\end{defi}
When $s=1$, we have $\frac{1}{K}\sum_{k=1}^K \vw_k\vw_k^\dagger = \int \vw\vw^\dagger d\vw = \frac{1}{d^n}\mId$, and thus $\mA_k = \frac{d^n}{K} \vw_k\vw_k^\dagger, k = 1,\dots, K$ form a rank-one POVM. For simplicity, we call such an induced POVM $\{\mA_k = \frac{d^n}{K} \vw_k\vw_k^\dagger\}$ as a $t$-design POVM. In \eqref{the definition of t designs}, we adopt the setting of uniform weights ($1/K$ for each $\vw_k$, which is mostly commonly used in practice) to simplify the analysis. However, a more general scenario with varied weights can also be considered \cite{dall2014accessible}. We note that $t$-designs always exist and can, in principle, be constructed in any dimension and for any $t$ \cite{seymour1984averaging, bajnok1992construction},  although in some cases these constructions can be inefficient, as they require vector sets of exponential size \cite{hayashi2005reexamination}. The following result establishes stable embeddings of any Hermitian matrices from SIC-POVM in terms of $\|\calA(\vrho) \|_2^2$.
\begin{theorem} [Quantity of spherical $t$-designs]
\label{l2 norm of 2-designs RIP}
Suppose that $\{\vw_k  \}_{k=1}^K\subset \C^{d^n}$ constitutes a spherical $t$-design ($t\ge 2$). Let $\calA$ be the linear map in \eqref{The defi of population measurement in Q cases (K measurements)} corresponding to the induced POVM $\{\mA_k = \frac{d^n}{K} \vw_k\vw_k^\dagger\}$. Then for arbitrary Hermitian matrix $\vrho\in\C^{d^n\times d^n}$, $\calA(\vrho)$ satisfies
\begin{eqnarray}
\label{The l2 norm of A(rho) 2_designs}
\|\calA(\vrho) \|_2^2 = \sum_{k=1}^K \frac{d^{2n}}{K^2}\<\vw_k\vw_k^\dagger, \vrho \>^2 = \frac{d^n(\|\vrho\|_F^2 + (\trace(\vrho))^2)}{K(d^n + 1)}\geq \frac{d^n\|\vrho\|_F^2}{K(d^n + 1)}.
\end{eqnarray}
\end{theorem}
The proof for this is presented in \Cref{Proof of RIP 2 designs in Appe}. Note that uniformly distributed $t$-designs ($t \geq 2$) are informationally complete; unlike Haar random projective measurements, only one POVM ($Q = 1$) is required to determined the quantity of $\|\calA(\vrho) \|_2^2$.

\paragraph{Projective Rank-one Measurements--Haar random projective measurements}

While tight rank-one IC-POVMs provide a theoretically optimal framework, they are generally challenging to implement in practice.
This motivates the use of projective rank-one measurements, which are experimentally more feasible,  where each measurement is associated with an arbitrary basis $\{\vphi_k\}$.  Since our focus is on projective measurements in random bases, we primarily consider Haar random projective measurements. Specifically, we define a unitary matrix $\mU_1 = \begin{bmatrix} \vphi_{1,1} & \cdots & \vphi_{1,d^n} \end{bmatrix}\in\C^{d^n\times d^n}$ and apply $\mU_1$ to the quantum state $\vrho$ prior to performing the projective measurement in a physically convenient basis (denoted as $\{\ve_k\}$), where $\mU_1 \ve_k = \vphi_{1,k}$ for any $k\leq d^n$. This process is mathematically expressed as $\innerprod{\mA_{1,k}}{\vrho} =\<\vphi_{1,k}\vphi_{1,k}^\dagger, \vrho  \> =  \ve_k^\dagger (\mU_1^{\dagger} \vrho \mU_1) \ve_k$. However, as previously mentioned, a single projective measurement (a single POVM) is insufficient to fully reconstruct a general quantum state $\vrho$. Therefore, projective measurements are performed in multiple bases, or more generally, using multiple POVMs, to obtain complete information about the quantum state. Specifically, we generate $Q$ unitary matrices  $\mU_i = \begin{bmatrix} \vphi_{i,1} & \cdots & \vphi_{i,d^n} \end{bmatrix}\in\C^{d^n\times d^n}, i = 1,\dots, Q$ and form population measurements as follows:
\begin{eqnarray}
    \label{Probability of rank-one POVM Haar}
    \innerprod{\mA_{i,k}}{\vrho} = \innerprod{\vphi_{i,k} \vphi_{i,k}^\dagger}{\vrho} = \vphi_{i,k}^\dagger \vrho \vphi_{i,k}, i\in[Q], k\in[d^n].
\end{eqnarray}

Next, we focus on determining the required number of POVMs to enable accurate quantum state recovery. A critical property in this context is the concept of a \emph{stable embedding}, which plays a pivotal role in recovery problems for sparse signals and low-rank matrices/tensors~\cite{donoho2006compressed, candes2006robust, candes2008introduction, recht2010guaranteed,candes2011tight, eftekhari2015new, Rauhut17, qin2024guaranteed, qin2024quantum}. Extensive analysis has been conducted on the low-dimensional structure of signals of interest to understand stable embedding. Leveraging Haar random projective measurements, we use modified Mendelson’s small ball method \Cref{Small Ball Method_abs_New} in {Appendix} \ref{Auxiliary materials} to demonstrate the existence of a stable embedding, which corresponds to deriving a lower bound for a nonnegative empirical process.
\begin{theorem} [Stable embedding of multiple Haar random projective measurements]
\label{Small_Method_ROHaar_Measurement PEPO}
Let $\calA:\C^{d^{n}\times d^{n}} \rightarrow \R^{Kd^n}$ be the linear map defined in \eqref{The defi of population measurement in Q cases (K measurements)} that is induced by $Q$ random unitary matrices. Assume that
\begin{eqnarray}
    \label{number of sample rank one haar PEPO}
    Q\geq \Omega(D),
\end{eqnarray}
in which $D = \begin{cases}
    \sum_{a=1}^{q}\sum_{b = 1}^{p} d \, t_{\ul{a}\,\ul{b-1},\,\ul{a}\,\ul{b}} \, t_{\ul{a-1}\,\ul{b},\,\ul{a}\,\ul{b}} \, t_{\ul{a}\,\ul{b},\,\ul{a}\,\ul{b+1}} \, t_{\ul{a}\,\ul{b},\,\ul{a+1}\,\ul{b}} \, \log(1+qp), & \text{PEPS} \\
     \sum_{a=1}^{q}\sum_{b = 1}^{p} d^2 \, r_{\ul{a}\,\ul{b-1},\,\ul{a}\,\ul{b}} \, r_{\ul{a-1}\,\ul{b},\,\ul{a}\,\ul{b}} \, r_{\ul{a}\,\ul{b},\,\ul{a}\,\ul{b+1}} \, r_{\ul{a}\,\ul{b},\,\ul{a+1}\,\ul{b}} \, \log(1+qp), & \text{PEPO}
    \end{cases}$
with  $qp = n$, $t_{\ul{a}\,\ul{0},\,\ul{a}\,\ul{1}} = t_{\ul{0}\,\ul{b},\,\ul{1}\,\ul{b}} = t_{\ul{a}\,\ul{p},\,\ul{a}\,\ul{p+1}} = t_{\ul{q}\,\ul{b},\,\ul{q+1}\,\ul{b}} = 1$ and
$r_{\ul{a}\,\ul{0},\,\ul{a}\,\ul{1}} = r_{\ul{0}\,\ul{b},\,\ul{1}\,\ul{b}} = r_{\ul{a}\,\ul{p},\,\ul{a}\,\ul{p+1}} = r_{\ul{q}\,\ul{b},\,\ul{q+1}\,\ul{b}} = 1$. Then with probability at least $1-e^{-\alpha_1Q}$ (where $\alpha_1$ is a positive constant.), $\calA$ obeys
\begin{eqnarray}
    \label{Small_Method_ROHaar_Measurement_Conclusion_Theorem PEPO}
    \|\calA(\vrho)\|_2^2  = \sum_{i=1}^Q\sum_{k=1}^{d^n} |\<\vphi_{i,k} \vphi_{i,k}^\dagger, \vrho  \>|^2   \geq \Omega\bigg(\frac{Q}{d^{n}}\|\vrho\|_F^2\bigg)
\end{eqnarray}
for all PEPSs $\vrho\in  \{\vrho\in\C^{d^{n}\times d^{n}}: \vrho = \vu \vu^\dagger, \vu = [\mU_{\ul{a}\, \ul{b}}]_{a=1,b=1}^{q,p},  \text{rank}(\vu) = \mT \}$ and all PEPOs $\vrho\in  \{\vrho\in\C^{d^{n}\times d^{n}}: \vrho = \vrho^\dagger, \vrho = [\mX_{\ul{a}\, \ul{b}}]_{a=1,b=1}^{q,p},  \text{rank}(\vrho) = \mR \}$.
\end{theorem}
The proof is given in {Appendix} \ref{Proof of Small ball PEPO}. This inference further indicates stable recovery. For any two PEPOs $\vrho_1$ and $\vrho_2$, we have
\begin{eqnarray}
\label{stable embedding of difference}
\|\calA(\vrho_1 - \vrho_2)\|_2^2  \geq \Omega\bigg(\frac{Q}{d^{n}}\|\vrho_1 - \vrho_2\|_F^2\bigg),
\end{eqnarray}
which ensures distinct measurements (i.e., $\calA(\vrho_1)\neq \calA(\vrho_2)$) as long as $\vrho_1 \neq \vrho_2$. In addition, observing the requirements of $Q$ in \eqref{number of sample rank one haar PEPO}, it is notable that we do not capitalize on the randomness among different columns within a random unitary matrix. Consequently, the necessity for a relatively large number of POVMs emerges. However, given the weak local correlations between columns in the unitary matrix, owing to orthogonality being a global property \cite{tropp2012comparison}, we posit that the requirement on $Q$ could be substantially mitigated, perhaps even to just $Q=1$. Indeed, as per \cite[Theorem 3]{jiang2006many}, when $n\to \infty$, in an ``in probability" sense, all elements (scaled by $\sqrt{d^{n}}$) of approximately $o(\frac{d^{n}}{n\log d})$ columns in a Haar-distributed random unitary matrix can be closely approximated by entries generated independently from a standard complex normal distribution. Rigorously leveraging this independence could ensure stable embedding \eqref{Small_Method_ROHaar_Measurement PEPO} with just a single POVM, analogous to the case where independent columns from a multivariate complex normal distribution suffice for the stable embedding of the MPO \cite[Theorem 2]{qin2024quantum}, which can be seen as a special case of the PEPO.

\section{Stable Recovery for Empirical Measurements}
\label{sec: stable recovery}
In Section \ref{sec: stable embedding of PEPO}, we define a distinct set of population measurements $\calA(\vrho)$ for any ground-truth PEPS or PEPO $\vrho^\star$. Expanding on these insights, we delve into the stable recovery of the PEPS- or PEPO-based density operator $\vrho^\star$ from empirical measurements obtained through diverse measurement methodologies. With empirical measurements $\wh \vp$, for simplicity, we consider minimizing the following constrained least squares objective:
\begin{eqnarray}
    \label{The loss function in QST}
    \wh{\vrho} =
\argmin_{\vrho\in\setX}\|\calA(\vrho) - {\widehat{\bm p}}\|_2^2,
\end{eqnarray}
where $\setX$ denotes the set of PEPSs or PEPOs on the $q\times p$ lattice:
\begin{eqnarray}
    \label{The set of the PEPS}
\text{PEPSs: }    \setX_{\mT} &\!\!\!\!=\!\!\!\!& \{\vrho\in\C^{d^{n}\times d^{n}}: \vrho = \vu \vu^\dagger, \|\vu\|_2=1, n = qp, \vu = [\mU_{\ul{a}\, \ul{b}}]_{a=1,b=1}^{q,p},  \text{rank}(\vu) = \mT\},\\
    \label{The set of the PEPO}
 \text{PEPOs: }   \setX_{\mR} &\!\!\!\!=\!\!\!\!& \{\vrho\in\C^{d^{n}\times d^{n}}: \vrho = \vrho^\dagger, \trace(\vrho) = 1, n = qp, \vrho = [\mX_{\ul{a}\, \ul{b}}]_{a=1,b=1}^{q,p}, \text{rank}(\vrho) = \mR  \}.
\end{eqnarray}
Since $\wh \vrho$ is a global solution to \eqref{The loss function in QST} and  $\vrho^\star \in \setX$, we have
\begin{eqnarray}
    \label{whrho and rho star relationship}
    \hspace{-0.1cm}0 &\!\!\!\!\!\leq\!\!\!\!\!& \|\calA(\vrho^\star) - \wh{\vp} \|_2^2  - \|\calA(\wh{\vrho}) - \wh{\vp} \|_2^2\nonumber\\
&\!\!\!\!\!=\!\!\!\!\!&\|\calA(\vrho^\star)-\calA(\vrho^\star) - \veta\|_2^2 - \|\calA(\wh{\vrho})-\calA(\vrho^\star) - \veta\|_2^2\nonumber\\
&\!\!\!\!\!=\!\!\!\!\!& 2\<\calA(\vrho^\star)+\veta, \calA(\wh{\vrho} - \vrho^\star)  \> + \|\calA(\vrho^\star)\|_2^2 - \|\calA(\wh{\vrho})\|_2^2\nonumber\\
&\!\!\!\!\!=\!\!\!\!\!& 2\<  \veta, \calA(\wh{\vrho} - \vrho^\star) \> - \|\calA(\wh{\vrho} - \vrho^\star)\|_2^2,
\end{eqnarray}
where we denote by $\veta$ the measurement error as $    \veta = \wh{\vp} -  \vp = \wh{\vp} - \calA(\vrho^\star) =  \begin{bmatrix}
          \veta_{1}^\top,
          \cdots,
          \veta_{Q}^\top\end{bmatrix}^\top$ with $\eta_{i,k}$ being the $k$-th element in $\veta_i$.
This further implies that
\begin{eqnarray}
    \label{whrho and rho^star relationship_1}
    \|\calA(\wh{\vrho} - \vrho^\star)\|_2^2 \leq 2\<  \veta, \calA(\wh{\vrho} - \vrho^\star) \>.
\end{eqnarray}
According to\Cref{l2 norm of 2-designs RIP} and \Cref{Small_Method_ROHaar_Measurement PEPO}, the left-hand side of the above equation can be further lower bounded by
\begin{eqnarray}
    \label{summary of lower bound of difference PEPO}
    \|\calA(\wh{\vrho} - \vrho^\star)\|_2^2 \geq \begin{cases}
    \frac{d^{n}\|\wh\vrho - \vrho^\star\|_F^2}{K(d^{n} + 1)}, &  \text{$t$-designs, $t\geq 2$}, \\
    \Omega\big(\frac{Q\|\wh\vrho - \vrho^\star\|_F^2}{d^{n}}\big), &  \text{Haar}.
\end{cases}
\end{eqnarray}
For the right-hand side of \eqref{whrho and rho^star relationship_1}, %presents a challenge. A direct application of the Cauchy-Schwarz inequality $\<  \veta, \calA(\wh{\vrho} - \vrho^\star) \> \le \|\veta\|_2 \cdot \|\calA(\wh{\vrho} - \vrho^\star)\|_2$ lacks precision due to the scaling of $\|\veta\|_2$ as $\frac{1}{\sqrt{M}}$ according to \cite[eq.(35)]{qin2024quantum}. Instead,
we exploit the randomness of $\veta$ and utilize the concentration bound presented in \Cref{General bound of multinomial distribution Q cases} in {Appendix} \ref{Auxiliary materials} for multinomial random variables. Building on this, we separately analyze $t$-design POVMs and Haar random projective measurements for PEPS and PEPO, and present the main findings as follows:

\begin{theorem} [Stable recovery under $t$-design POVMs ($t\geq 3$) for PEPSs]
\label{final conclusion of 3-designs theorem PEPS}
Suppose $\{\mA_1, \ldots, \mA_{K}\}$ form a set of $t$-design POVMs ($t\geq 3$) and $\vrho^\star\in\C^{d^n \times d^n}$ ($n = qp$) is a PEPS state with ranks $\mT$.  We measure the state $M$ times using the POVM, and specifically, for any $\delta>0$, suppose
\begin{eqnarray}
    \label{upper bound of_QM 3 designs PEPS}
    M \geq \Omega\bigg(\frac{D_{\text{PEPS}}}{\delta^2}\bigg),
\end{eqnarray}
where $D_{\text{PEPS}} = \sum_{a=1}^{q}\sum_{b = 1}^{p} d \,
t_{\ul{a}\,\ul{b-1},\,\ul{a}\,\ul{b}} \,
t_{\ul{a-1}\,\ul{b},\,\ul{a}\,\ul{b}} \,
t_{\ul{a}\,\ul{b},\,\ul{a}\,\ul{b+1}} \,
t_{\ul{a}\,\ul{b},\,\ul{a+1}\,\ul{b}}
\log(1 + qp)$ with $t_{\ul{a}\,\ul{0},\,\ul{a}\,\ul{1}} =
t_{\ul{0}\,\ul{b},\,\ul{1}\,\ul{b}} =
t_{\ul{a}\,\ul{p},\,\ul{a}\,\ul{p+1}} =
t_{\ul{q}\,\ul{b},\,\ul{q+1}\,\ul{b}} = 1$.
Then, with probability $1 - e^{- \alpha_2 D_{\text{PEPS}}}$, the solution $\wh \vrho$ of the constrained least squares \eqref{The loss function in QST} satisfies
\begin{eqnarray}
\label{final conclusion of recovery error1 for t-design larger 3 PEPS}
\|\wh{\vrho} - \vrho^\star\|_1\leq  \delta,
\end{eqnarray}
where $\alpha_2$ is a positive constant.
\end{theorem}

\begin{theorem} [Stable recovery under $t$-design POVMs ($t\geq 3$) for PEPOs]
\label{final conclusion of 3-designs theorem}
Suppose $\{\mA_1, \ldots, \mA_{K}\}$ form a set of $t$-design POVMs ($t\geq 3$) and $\vrho^\star\in\C^{d^n \times d^n}$ ($n = qp$) is a PEPO state with ranks $\mR$.  We measure the state $M$ times using the POVM, and specifically, for any $\epsilon>0$, suppose
\begin{eqnarray}
    \label{upper bound of_QM 3 designs}
    M \geq \Omega\bigg(\frac{D_{\text{PEPO}}}{\epsilon^2}\bigg),
\end{eqnarray}
where $D_{\text{PEPO}} = \sum_{a=1}^{q}\sum_{b = 1}^{p} d^2 \,
r_{\ul{a}\,\ul{b-1},\,\ul{a}\,\ul{b}} \,
r_{\ul{a-1}\,\ul{b},\,\ul{a}\,\ul{b}} \,
r_{\ul{a}\,\ul{b},\,\ul{a}\,\ul{b+1}} \,
r_{\ul{a}\,\ul{b},\,\ul{a+1}\,\ul{b}}
\log(1 + qp)$ with $r_{\ul{a}\,\ul{0},\,\ul{a}\,\ul{1}} =
r_{\ul{0}\,\ul{b},\,\ul{1}\,\ul{b}} =
r_{\ul{a}\,\ul{p},\,\ul{a}\,\ul{p+1}} =
r_{\ul{q}\,\ul{b},\,\ul{q+1}\,\ul{b}} = 1$.
Then, with probability $1 - e^{- \alpha_3 D_{\text{PEPO}}}$, the solution $\wh \vrho$ of the constrained least squares \eqref{The loss function in QST} satisfies
\begin{eqnarray}
\label{final conclusion of recovery error1 for t-design larger 3}
\|\wh{\vrho} - \vrho^\star\|_F\leq  \epsilon,
\end{eqnarray}
where $\alpha_3$ is a positive constant.
\end{theorem}
The proof is provided in {Appendix} \ref{Proof of 2designs M}. \Cref{final conclusion of 3-designs theorem PEPS,final conclusion of recovery error1 for t-design larger 3} demonstrate that the number of state copies scales linearly as $O(n)$, representing the optimal scenario relative to the intrinsic degrees of freedom: $O(nd \|\mT\|_{\infty}^4 )$ for PEPS and $O(nd^2 \|\mR\|_{\infty}^4 )$ for PEPO. We emphasize that the analysis for MPOs \cite{noh2020efficient,qin2024quantum} cannot be straightforwardly generalized to PEPS or PEPO. This difficulty arises from the lack of a (quasi-)canonical form in PEPS/PEPO, which prevents the simultaneous selection of an orthonormal basis across all bond indices \cite{orus2019tensor}. Consequently, bounds are imposed on each tensor factor in PEPS and PEPO to facilitate the analysis, although these assumptions are purely technical and do not alter the required sample complexity. In addition, while spherical $t$-design POVMs ($t\geq 3$) provide strong guarantees in terms of sample complexity, they remain mostly of theoretical interest and are difficult to realize experimentally. In contrast, Haar random projective measurements are not only analytically tractable but also more amenable to physical implementation, making them a natural candidate for further analysis. Specifically, we have

\begin{theorem} [Stable recovery under Haar random projective measurements for PEPSs]
\label{final conclusion of Haar theorem PEPS}
Given a PEPS state $\vrho^\star\in\C^{d^n \times d^n}$ ($n = qp$) with PEPO ranks $\mT$, independently generate $Q$ Haar-distributed random unitary matrices $[  \vphi_{i,1} \,  \cdots \, $ $ \vphi_{i,d^n}], i=1,\dots,Q$. Use each induced rank-one POVM $\{\vphi_{i,k}\vphi_{i,k}^\dagger\}_{k=1}^{d^n}$ to measure the state $M$ times and get the empirical measurements $\wh \vp_i$.
For any $\delta>0$, suppose  $Q\geq \Omega(D_{\text{PEPS}})$ and
\begin{eqnarray}
    \label{upper bound of_QM Haar PEPS}
    QM \geq \Omega\bigg(\frac{D_{\text{PEPS}} (\log Q + n \log d)^2}{\delta^2}\bigg)
\end{eqnarray}
where $D_{\text{PEPS}} = \sum_{a=1}^{q}\sum_{b = 1}^{p} d \,
t_{\ul{a}\,\ul{b-1},\,\ul{a}\,\ul{b}} \,
t_{\ul{a-1}\,\ul{b},\,\ul{a}\,\ul{b}} \,
t_{\ul{a}\,\ul{b},\,\ul{a}\,\ul{b+1}} \,
t_{\ul{a}\,\ul{b},\,\ul{a+1}\,\ul{b}}
\log(1 + qp)$ with $t_{\ul{a}\,\ul{0},\,\ul{a}\,\ul{1}} =
t_{\ul{0}\,\ul{b},\,\ul{1}\,\ul{b}} =
t_{\ul{a}\,\ul{p},\,\ul{a}\,\ul{p+1}} =
t_{\ul{q}\,\ul{b},\,\ul{q+1}\,\ul{b}} = 1$.
Then any global solution $\wh \vrho$ of \eqref{The loss function in QST} satisfies
\begin{eqnarray}
    \label{Statistical Error_Haar_Measurement_Conclusion_Theorem PEPS}
    \|\wh{\vrho}-\vrho^\star\|_1 \le \delta
\end{eqnarray}
with probability at least $\min\{1-e^{-\alpha_4 (\log Q + n)} - e^{-\alpha_5 D_{\text{PEPS}}}, 1-e^{-\alpha_1Q} \}$, where $\alpha_4 $ and $\alpha_5$ are positive constants,  $\alpha_1$  corresponds to constants of the probability in \Cref{Small_Method_ROHaar_Measurement PEPO}.

\end{theorem}

\begin{theorem} [Stable recovery under Haar random projective measurements for PEPOs]
\label{final conclusion of Haar theorem}
Given a PEPO state $\vrho^\star\in\C^{d^n \times d^n}$ ($n = qp$) with PEPO ranks $\mR$, independently generate $Q$ Haar-distributed random unitary matrices $[  \vphi_{i,1} \,  \cdots \, $ $ \vphi_{i,d^n}], i=1,\dots,Q$. Use each induced rank-one POVM $\{\vphi_{i,k}\vphi_{i,k}^\dagger\}_{k=1}^{d^n}$ to measure the state $M$ times and get the empirical measurements $\wh \vp_i$.
For any $\epsilon>0$, suppose  $Q\geq \Omega(D_{\text{PEPO}})$ and
\begin{eqnarray}
    \label{upper bound of_QM Haar}
    QM \geq \Omega\bigg(\frac{D_{\text{PEPO}} (\log Q + n \log d)^2}{\epsilon^2}\bigg)
\end{eqnarray}
where $D_{\text{PEPO}} = \sum_{a=1}^{q}\sum_{b = 1}^{p} d^2 \,
r_{\ul{a}\,\ul{b-1},\,\ul{a}\,\ul{b}} \,
r_{\ul{a-1}\,\ul{b},\,\ul{a}\,\ul{b}} \,
r_{\ul{a}\,\ul{b},\,\ul{a}\,\ul{b+1}} \,
r_{\ul{a}\,\ul{b},\,\ul{a+1}\,\ul{b}}
\log(1 + qp)$ with $r_{\ul{a}\,\ul{0},\,\ul{a}\,\ul{1}} =
r_{\ul{0}\,\ul{b},\,\ul{1}\,\ul{b}} =
r_{\ul{a}\,\ul{p},\,\ul{a}\,\ul{p+1}} =
r_{\ul{q}\,\ul{b},\,\ul{q+1}\,\ul{b}} = 1$.
Then any global solution $\wh \vrho$ of \eqref{The loss function in QST} satisfies
\begin{eqnarray}
    \label{Statistical Error_Haar_Measurement_Conclusion_Theorem}
    \|\wh{\vrho}-\vrho^\star\|_F \le \epsilon
\end{eqnarray}
with probability at least $\min\{1-e^{-\alpha_6 (\log Q + n)} - e^{-\alpha_7 D_{\text{PEPO}}}, 1-e^{-\alpha_1Q} \}$, where $\alpha_6 $ and $\alpha_7$ are positive constants,  $\alpha_1$  corresponds to constants of the probability in \Cref{Small_Method_ROHaar_Measurement PEPO}.

\end{theorem}
The detailed proof is shown in {Appendix} \ref{Proof of final conclusion of Haar theorem}. \Cref{final conclusion of Haar theorem PEPS,final conclusion of Haar theorem} guarantees stable recovery of the ground-truth state when the total number of state copies $QM$ grows only polynomially with respect to the number of qudits $n$. It is noteworthy that \eqref{upper bound of_QM Haar PEPS} and \eqref{upper bound of_QM Haar} impose only a requirement for $QM$ to be sufficiently large, without any constraints on the number of measurement times $M$ for each POVM. Therefore, \Cref{Statistical Error_Haar_Measurement_Conclusion_Theorem PEPS,Statistical Error_Haar_Measurement_Conclusion_Theorem} provides theoretical support for the practical application of single-shot measurements (i.e., $M = 1$ where each POVM is measured only once). However, the orders  $n^3$ of the polynomial in \eqref{upper bound of_QM Haar PEPS} and \eqref{upper bound of_QM Haar} are relatively large compared to the intrinsic degrees of freedom, namely $O(nd \|\mT\|_{\infty}^4 )$ for PEPS and $O(nd^2 \|\mR\|_{\infty}^4 )$ for PEPO. As discussed in {Section}~\ref{sec: stable embedding of PEPO}, this is attributed to the lack of utilization of independence between columns within a single unitary matrix. On the other hand, compared with spherical $t$-design POVMs ($t\geq 3$), Haar random projective measurements are not only mathematically tractable but also significantly easier to implement in quantum circuits, making them a more practical choice in experimental scenarios.

\paragraph*{Discussion} In accordance with the previous discussion, all recovery guarantees for PEPOs are established using the Frobenius norm rather than the trace norm. Nevertheless, if the PEPO state $\vrho^\star$ exhibits a low matrix rank, we can also establish a recovery guarantee in the trace norm by leveraging a bound between the trace distance and Hilbert-Schmidt distance for low-rank states, as proposed in \cite{coles2019strong}: $\|\wh\vrho - \vrho^\star\|_1 \leq 2\sqrt{\text{rank}(\vrho^\star)}\| \wh\vrho - \vrho^\star \|_F$. Although this approach may yield a trivial bound for high matrix ranks, we hypothesize that \Cref{final conclusion of Haar theorem,final conclusion of recovery error1 for t-design larger 3} can be extended to the trace norm, regardless of the matrix rank of $\vrho^\star$, through direct analysis. However, for physical PEPO states, the rank of these states is either 1, as observed in pure states such as the 2D cluster state, the ground state of the Toric Code model, the 2D resonating valence bond (RVB) state, and the 2D AKLT model, or they exhibit approximately low-rank behavior at low temperatures, as seen in thermal states and finite-temperature extensions of the Toric Code, RVB, and AKLT models.

It is worth noting that, in the estimation of PEPO states, although the solution $\wh{\vrho}$ of \eqref{The loss function in QST} may not adhere to physical constraints, we can introduce additional constraints to enforce physicality without compromising the recovery guarantee. Specifically, let $\vrho^\diamond$ denote the global solution to the following minimization problem with an additional PSD constraint:
\begin{eqnarray}
    \label{The loss function in QST-PSD}
    \vrho^\diamond = \argmin_{\vrho\in\setX_{\mR}, \vrho\succeq \vzero}\|\calA(\vrho) - \widehat\vp \|_2^2.
\end{eqnarray}
Then, $\vrho^\diamond$ enjoys the same guarantee as $\wh \vrho$ under the same setup as \Cref{final conclusion of Haar theorem,final conclusion of recovery error1 for t-design larger 3}. Alternatively, we can directly project $\wh \vrho$ onto the set of physical states $\setS_+:=\{\vrho\in\C^{d^n\times d^n}: \vrho \succeq \vzero, \trace(\vrho) = 1\}$ and denote $P_{\setS_+}$ as the projection onto the set $\setS_+$. Since $\setS_+$ is convex, the corresponding projector is non-expansive, thus satisfying:
\begin{eqnarray}
    \label{The nonexpansiveness of physical state}
    \|\vrho^\diamond-\vrho^\star\|_F =\|P_{\setS_+}(\wh{\vrho} )-\vrho^\star\|_F = \|P_{\setS_+}(\wh{\vrho} )-P_{\setS_+}(\vrho^\star)\|_F \le \|\wh{\vrho}-\vrho^\star\|_F \le \epsilon,
\end{eqnarray}
indicating that the projection step ensures the resulting state complies with physical constraints while maintaining or even enhancing the recovery guarantee.

Finally, our focus lies on PEPS and PEPO in the lattice structure in this paper, yet our conclusions can readily extend to any tensor network with a degree of freedom given by $d_{\text{tn}}$--corresponding to the number of independent parameters--and $n_{\text{tn}}$ tensor factors. By following the same analysis as presented in \Cref{final conclusion of Haar theorem,final conclusion of recovery error1 for t-design larger 3}, we can deduce that when
\begin{eqnarray}
    \label{upper bound of_QM for all cases in tensor network}
    M &\!\!\!\!=\!\!\!\!&  O (d_{\text{tn}} \log (1 + n_{\text{tn}} )/\epsilon^2),  \  \text{$t$-designs, $t\geq 3$},\\
     QM &\!\!\!\!=\!\!\!\!&  O(d_{\text{tn}} n^2 \log (1 + n_{\text{tn}} )/\epsilon^2) , \  \text{Haar},
\end{eqnarray}
we conclude that $\|\wh{\vrho} - \vrho^\star\|_F\leq  \epsilon$, where $\wh{\vrho}, \vrho^\star\in\C^{d^n\times d^n}$ respectively represent the estimated and ground-truth tensor networks.
This highlights that the accuracy of QST when using certain POVMs is contingent upon the structure of tensor networks. Therefore, to precisely recover a physical state, it is imperative to identify optimal tensor networks with minimal degrees of freedom.

\section{Conclusion}
\label{sec: conclusion}

In this paper, we investigate sampling bounds for recovering structured quantum states represented as projected entangled-pair states (PEPSs) projected entangled-pair operators (PEPOs). By analyzing quantum measurements based on spherical $t$-designs and Haar random projective measurements, we establish stable embedding properties that are essential for robust PEPS and PEPO recovery. Our results on PEPS- and PEPO-based quantum state tomography indicate that only a polynomial number of state copies relative to the qudit count is necessary for bounded recovery error of a PEPS or PEPO state. Our findings contribute to advancing quantum state tomography methods and lay a foundation for further research to refine our understanding of structured quantum states.

An important avenue for future research is the development of efficient optimization algorithms. Due to the lack of a canonical form in PEPSs and PEPOs, a common issue in looped tensor networks, estimating any PEPS or PEPO with optimal bond dimensions using the iterative hard thresholding method \cite{Rauhut17,qin2024sample} is not feasible. A summary of existing methods for obtaining PEPS or PEPO is provided in \cite[Section III-C]{orus2019tensor}; however, their theoretical performance, including error and convergence analyses, has not been addressed. A promising approach in PEPS- and PEPO-based QST is to optimize tensor factors via factorization techniques employing gradient descent \cite{qin2024guaranteed,jameson2024optimal}. The convergence and recovery error of such methods can be analyzed based on \cite{qin2024sample,qin2024guaranteed}. However, fine-tuning the estimated bond dimensions remains necessary, making the design of efficient factorization methods that adaptively update bond dimensions a critical direction for future exploration.

Another potential research direction is analyzing the required number of state copies using unitary $3$-designs. The measurements considered in this work are global measurements, as the unitary matrix rotates the entire system of qudits simultaneously, presenting challenges for practical implementation in quantum circuits. However, if the goal is solely to generate unitary $3$-designs, this can be achieved by sampling Clifford circuits uniformly at random \cite{kueng2015qubit,webb2015clifford,zhu2017multiqubit}. Since a unitary $3$-design POVM replicates the first $3$ moments of the full unitary group under the Haar measure, such POVMs share statistical properties with Haar random projective measurements. Consequently, we hypothesize that the required number of state copies for unitary $3$-design POVMs is polynomially proportional to the degrees of freedom of the PEPOs. This hypothesis, along with its verification, will be a focus of future work.

\section{Acknowledgments}
\label{sec: ack}

We acknowledge funding support from NSF Grants No. CCF-2241298 and ECCS-2409701.  ZQ gratefully acknowledges support from the MICDE Research Scholars Program at the University of Michigan.

\crefalias{section}{appendix}
\appendices

\section{Proof of \Cref{l2 norm of 2-designs RIP}}
\label{Proof of RIP 2 designs in Appe}

\begin{proof}
According to \cite[Theorem 1]{dall2014accessible} and \Cref{property of 2-designs}, we have
\begin{eqnarray}
\label{sum of 2-designs measurement}
\sum_{k=1}^K \frac{1}{K}(\vw_k\vw_k^\dagger)^{\otimes 2} = \frac{\mId + \mS }{d^n(d^n+1)},
\end{eqnarray}
where $\mS$ is the swap operator.  Multiplying both sides by $\frac{d^{2n}}{K}\vrho^{\otimes 2}$ and taking the trace we get
\begin{eqnarray}
\label{The l2 norm of A(rho) 2_designs proof}
\sum_{k=1}^K \frac{d^{2n}}{K^2}\trace((\vw_k\vw_k^\dagger)^{\otimes 2}(\vrho^{\otimes 2})  ) &\!\!\!\!=\!\!\!\!& \frac{d^n(\trace(\vrho^{\otimes 2}) + \trace(\mS\vrho^{\otimes 2}) )}{K(d^n+1)}\nonumber\\
&\!\!\!\!=\!\!\!\!&\frac{d^n((\trace(\vrho))^2 + \trace(\vrho^2)) }{K(d^n+1)}\nonumber\\
&\!\!\!\!=\!\!\!\!&\frac{d^n((\trace(\vrho))^2 + \|\vrho\|_F^2 )}{K(d^n+1)},\nonumber\\
\end{eqnarray}
where the penultimate line follows $\trace(\mS\vrho^{\otimes 2}) = \trace(\vrho^2)$ \cite[Lemma 17]{KuengACHA17} for any Hermitian matrix $\vrho$. Based on $\trace((\vw_k\vw_k^\dagger)^{\otimes 2}(\vrho^{\otimes 2})  ) = \trace( (\vw_k\vw_k^\dagger\vrho)\otimes (\vw_k\vw_k^\dagger\vrho) ) = (\trace(\vw_k\vw_k^\dagger\vrho))^2$, this completes the proof.
\end{proof}

\section{Proof of \Cref{Small_Method_ROHaar_Measurement PEPO}}
\label{Proof of Small ball PEPO}

\begin{proof}
In the proof, without loss of generality, we assume that all PEPSs $\vrho\in  \{\vrho\in\C^{d^{n}\times d^{n}}: \vrho = \vu \vu^\dagger, \vu = [\mU_{\ul{a}\, \ul{b}}]_{a=1,b=1}^{q,p},  \text{rank}(\vu) = \mT,\|\mU_{\ul{a}\, \ul{b}}\|_F \leq H_{ab}, a\in[q], b\in[p]  \}$ and all PEPOs $\vrho\in  \{\vrho\in\C^{d^{n}\times d^{n}}: \vrho = \vrho^\dagger, \vrho = [\mX_{\ul{a}\, \ul{b}}]_{a=1,b=1}^{q,p},  \text{rank}(\vrho) = \mR,\|\mX_{\ul{a}\, \ul{b}}\|_F \leq L_{ab}, a\in[q], b\in[p]  \}$. We emphasize that the constraints $\{H_{ab}\}_{a=1,b=1}^{q,p}$ and $\{L_{ab}\}_{a=1,b=1}^{q,p}$ are introduced solely to facilitate the theoretical analysis and do not affect the required sample complexity. Since PEPS can be regarded as a special case of PEPO, we first establish the result for PEPOs and then specialize it to the PEPS setting.

We define a set $\ddot\setX_{\mR} = \{\vrho\in\C^{d^{qp}\times d^{qp}}: \vrho = \vrho^\dagger, \vrho = [\mX_{\ul{a} \, \ul{b}}]_{a=1,b=1}^{q,p},  \text{rank}(\vrho) = \mR,  \|\mX_{\ul{a} \, \ul{b}}\|_F \leq L_{ab}, a\in[q], b\in[p]  \}$. Next, we prove \Cref{Small_Method_ROHaar_Measurement PEPO} using the modified Mendelson's small ball method. Let $\{ \vphi_{1},\dots, \vphi_{K} \}$ be the first $K$ columns of a randomly generated Haar distributed unitary matrix, and let $\{ \vphi_{i,1},\dots, \vphi_{i,K} \}_{i=1}^Q$ be independent copies of $\{ \vphi_{1},\dots, \vphi_{K} \}$. According to \Cref{Small Ball Method_abs_New},  we need to  bound
\begin{eqnarray} \label{ProofOfSmallBallMethod_For_Haar_Measurement PEPO_H}
    H_{\xi}(\ddot\setX_{\mR}) = \inf_{\vrho\in\ddot\setX_{\mR}} \frac{1}{K}\sum_{k=1}^K\mathbb{P} \{|\< \vphi_k\vphi_k^\dagger,\vrho \>|\geq \xi  \}
\end{eqnarray}
and
\begin{eqnarray}
    \label{ProofOfSmallBallMethod_For_Haar_Measurement PEPO_W}
    W(\ddot\setX_{\mR}) = \E \sup_{\vrho\in \ddot\setX_{\mR}}\frac{1}{\sqrt{QK}}\sum_{i=1}^Q\sum_{k=1}^K \<\epsilon_i\vphi_{i,k}\vphi_{i,k}^\dagger, \vrho  \>,
\end{eqnarray}
where $\epsilon_i,i=1,\dots,Q$ are independent Rademacher random variables. Below we study the two quantities separately.

According to \cite[eq. (105)]{qin2024quantum}, we can directly obtain
\begin{eqnarray}
\label{ProofOfSmallBallMethod_For_Haar_Measurement PEPO_H upper bound}
    H_{\xi}(\ddot\setX_{\mR}) \ge c_0, \ \forall \xi \le \frac{c_1}{d^n}\|\vrho\|_F,
\end{eqnarray}
where $c_0, c_1$ are positive constants.

Next, we apply a covering argument to analyze \eqref{ProofOfSmallBallMethod_For_Haar_Measurement PEPO_W}. According to \cite{zhang2018tensor}, we can construct $\epsilon_{ab}$-net $\{ \mX_{\ul{a}\,\ul{b}}^{(1)}, \dots, \mX_{\ul{a}\,\ul{b}}^{(n_{\ul{a}\,\ul{b}})}\}$ with the covering number
\begin{eqnarray}
    \label{covering number of one factor}
    n_{\ul{a}\,\ul{b}}&\!\!\!\! \leq \!\!\!\!& \bigg(\frac{2L_{ab}+L_{ab}\epsilon_{ab}}{L_{ab}\epsilon_{ab}}\bigg)^{ d^2 r_{\ul{a}\,\ul{b-1},\,\ul{a}\,\ul{b}}\;
    r_{\ul{a-1}\,\ul{b},\,\ul{a}\,\ul{b}}\;
    r_{\ul{a}\,\ul{b},\,\ul{a}\,\ul{b+1}}\;
    r_{\ul{a}\,\ul{b},\,\ul{a+1}\,\ul{b}}}\nonumber\\
    &\!\!\!\! = \!\!\!\!&\bigg(\frac{2+\epsilon_{ab}}{\epsilon_{ab}}\bigg)^{ d^2 r_{\ul{a}\,\ul{b-1},\,\ul{a}\,\ul{b}}\;
    r_{\ul{a-1}\,\ul{b},\,\ul{a}\,\ul{b}}\;
    r_{\ul{a}\,\ul{b},\,\ul{a}\,\ul{b+1}}\;
    r_{\ul{a}\,\ul{b},\,\ul{a+1}\,\ul{b}}}
\end{eqnarray}
for $\{\mX_{\ul{a}\,\ul{b}} \in\C^{d \times d \times r_{\ul{a}\,\ul{b-1},\,\ul{a}\,\ul{b}} \times
    r_{\ul{a-1}\,\ul{b},\,\ul{a}\,\ul{b}} \times
    r_{\ul{a}\,\ul{b},\,\ul{a}\,\ul{b+1}} \times
    r_{\ul{a}\,\ul{b},\,\ul{a+1}\,\ul{b}}} :
    \|\mX_{\ul{a}\,\ul{b}}\|_F\leq L_{ab}  \}$ such that
\begin{eqnarray}
    \label{requirement of one factor}
    \sup_{\mX_{\ul{a}\,\ul{b}}: \|\mX_{\ul{a}\,\ul{b}}\|_F\leq L_{ab}}\min_{h_{\ul{a}\,\ul{b}}\leq n_{\ul{a}\,\ul{b}}} \|\mX_{\ul{a}\,\ul{b}}-\mX_{\ul{a}\,\ul{b}}^{(h_{\ul{a}\,\ul{b}})}\|_F\leq L_{ab}\epsilon_{ab}.
\end{eqnarray}
Then we define the set $\widetilde\setX_{\mR} = \{\vrho^{(h)}: \vrho^{(h)} = {\vrho^{(h)}}^\dagger, \vrho^{(h)} = [\mX_{\ul{a}\,\ul{b}}^{(h_{\ul{a}\,\ul{b}})}]_{a=1,b=1}^{q,p}, \|\vrho^{(h)}\|_F =\|\vrho\|_F, \rho\in\ddot\setX_{\mR}, \text{rank}(\vrho^{(h)}) = \mR, \|\mX_{\ul{a}\,\ul{b}}^{(h_{\ul{a}\,\ul{b}})}\|_F \leq L_{ab},  1\leq h_{\ul{a}\,\ul{b}}\leq n_{\ul{a}\,\ul{b}}, a\in[q], b\in[p]  \}\subset \ddot\setX_{\mR}$ which obeys
\begin{eqnarray}
    \label{covering number for the PEPO subset}
    |\widetilde\setX_{\mR} | \leq \Pi_{a = 1}^{q}\Pi_{b = 1}^{p} \bigg(\frac{2+\epsilon_{ab}}{\epsilon_{ab}}\bigg)^{ d^2 r_{\ul{a}\,\ul{b-1},\,\ul{a}\,\ul{b}}\;
    r_{\ul{a-1}\,\ul{b},\,\ul{a}\,\ul{b}}\;
    r_{\ul{a}\,\ul{b},\,\ul{a}\,\ul{b+1}}\;
    r_{\ul{a}\,\ul{b},\,\ul{a+1}\,\ul{b}}}.
\end{eqnarray}

Denote by
\begin{align*}
    \wh\vrho &:= \arg\sup_{\vrho\in\ddot\setX_{\mR}} \frac{1}{\sqrt{QK}}\sum_{i=1}^Q\sum_{k=1}^K \<\epsilon_i\vphi_{i,k}\vphi_{i,k}^\dagger, \vrho  \>,\\
    T & := \frac{1}{\sqrt{QK}}\sum_{i=1}^Q\sum_{k=1}^K \<\epsilon_i\vphi_{i,k}\vphi_{i,k}^\dagger, \wh\vrho  \>.
\end{align*}

Now taking $\epsilon_{ab} = \frac{1}{2 qp}$ gives
\begin{align} \label{ProofOf<H,X>PEPO upper}
    T&=\frac{1}{\sqrt{QK}}\sum_{i=1}^Q\sum_{k=1}^K \<\epsilon_i\vphi_{i,k}\vphi_{i,k}^\dagger, \wh\vrho -  \vrho^{(h)} + \vrho^{(h)} \>\nonumber\\
    & \leq \frac{1}{\sqrt{QK}}\sum_{i=1}^Q\sum_{k=1}^K \<\epsilon_i\vphi_{i,k}\vphi_{i,k}^\dagger,  \vrho^{(h)} \>\nonumber\\
    & + \frac{1}{\sqrt{QK}}\sum_{i=1}^Q\sum_{k=1}^K\sum_{a=1}^{q}\sum_{b=1}^{p}\<\epsilon_i\vphi_{i,k}\vphi_{i,k}^\dagger, [\mX_{\ul{1} \, \ul{1}},\dots, \mX_{\ul{a} \, \ul{b}} - \mX^{(h_{\ul{a} \, \ul{b}})}_{\ul{a} \, \ul{b}},\dots,  \mX^{(h_{\ul{q} \, \ul{p}})}_{\ul{q} \, \ul{p}}   ]  \>\nonumber\\
    & \leq \frac{1}{\sqrt{QK}}\sum_{i=1}^Q\sum_{k=1}^K \<\epsilon_i\vphi_{i,k}\vphi_{i,k}^\dagger,  \vrho^{(h)} \>\nonumber\\
    & +\sum_{a=1}^{q}\sum_{b=1}^{p} L_{ab}\epsilon_{ab}\frac{1}{\sqrt{QK}}\sum_{i=1}^Q\sum_{k=1}^K\<\epsilon_i\vphi_{i,k}\vphi_{i,k}^\dagger, [\mX_{\ul{1} \, \ul{1}},\dots, \frac{\mX_{\ul{a} \, \ul{b}} - \mX^{(h_{\ul{a} \, \ul{b}})}_{\ul{a} \, \ul{b}}}{\|\mX_{\ul{a} \, \ul{b}} - \mX^{(h_{\ul{a} \, \ul{b}})}_{\ul{a} \, \ul{b}}\|_F},\dots,  \mX^{(h_{\ul{q} \, \ul{p}})}_{\ul{q} \, \ul{p}}   ]  \> \nonumber\\
    & \leq \frac{1}{\sqrt{QK}}\sum_{i=1}^Q\sum_{k=1}^K \<\epsilon_i\vphi_{i,k}\vphi_{i,k}^\dagger,  \vrho^{(h)} \> + \sum_{a=1}^{q}\sum_{b=1}^{p} \epsilon_{ab} T\nonumber\\
    & =  \frac{1}{\sqrt{QK}}\sum_{i=1}^Q\sum_{k=1}^K \<\epsilon_i\vphi_{i,k}\vphi_{i,k}^\dagger,  \vrho^{(h)} \> + \frac{T}{2}.
\end{align}
Here, we expand $\wh\vrho - \vrho^{(h)}$ into $qp$ terms and $[\mX_{\ul{1} \, \ul{1}},\dots, \mX_{\ul{a} \, \ul{b}} - \mX^{(h_{\ul{a} \, \ul{b}})}_{\ul{a} \, \ul{b}},\dots,  \mX^{(h_{\ul{q} \, \ul{p}})}_{\ul{q} \, \ul{p}} ]$ denotes the PEPO format. The third inequality follows from the fact that $L_{ab}\|\frac{\mX_{\ul{a} \, \ul{b}} - \mX^{(h_{\ul{a} \, \ul{b}})}_{\ul{a} \, \ul{b}}}{\|\mX_{\ul{a} \, \ul{b}} - \mX^{(h_{\ul{a} \, \ul{b}})}_{\ul{a} \, \ul{b}}\|_F}\|_F= L_{ab}$.

Since each $\< \vphi_{i,k}\vphi_{i,k}^\dagger,\vrho \>$ is a subexponetial random variable with $\|\< \vphi_{i,k}\vphi_{i,k}^\dagger,\vrho \>\|_{\psi_1} = O(\frac{\|\vrho\|_F}{d^n})$ according to \cite[eq. (103)]{qin2024quantum}, $\epsilon_i\vphi_{i,k}^\dagger\vrho\vphi_{i,k}$ is a centered subexponential random variable with the subexponential norm $\|\epsilon_i\vphi_{i,k}^\dagger\vrho\vphi_{i,k}\|_{\psi_1}= O(\frac{\|\vrho\|_F}{d^n})$. On the other hand, for any $i$, the random vectors $\vphi_{i,k}$ and $\vphi_{i,k'}$ are not dependent to each other for $k\neq k'$. Thus, we use \Cref{Concentration inequality of sum of dependent sub-exp sum} to obtain its concentration inequality as follows:
\begin{eqnarray}
    \label{ProofOfSmallBallMethod_For_Haar_Measurement_W_1}
     &\!\!\!\!\!\!\!\!&\P{ T \geq t}\nonumber\\
     &\!\!\!\! \leq \!\!\!\!&\P{ \frac{1}{\sqrt{QK}}\sum_{i=1}^Q\sum_{k=1}^K \<\epsilon_{i}\vphi_{i,k}\vphi_{i,k}^\dagger, \vrho^{(h)}  \> \geq \frac{t}{2}}\nonumber\\
     &\!\!\!\!\leq\!\!\!\!& \begin{cases}
    \bigg(4qp+1\bigg)^{\ol D_{\text{PEPO}}  }e^{-\frac{c_2 d^{2n}t^2}{4K\|\vrho\|_F^2}}, & t\leq \frac{c_4\sqrt{QK} }{d^n}\|\vrho\|_F\\
    \bigg(4qp+1\bigg)^{\ol D_{\text{PEPO}}  }e^{-\frac{c_3 \sqrt{Q} d^nt}{2\sqrt{K}\|\vrho\|_F}}, & t > \frac{c_4\sqrt{QK} }{d^n}\|\vrho\|_F
  \end{cases} \nonumber\\
    &\!\!\!\!\leq\!\!\!\!&\begin{cases}
    e^{-\frac{c_2 d^{2n}t^2}{4K\|\vrho\|_F^2} +C\ol D_{\text{PEPO}}\log(1+qp)}, & t\leq \frac{c_4\sqrt{QK} }{d^n}\|\vrho\|_F\\
    e^{-\frac{c_3 \sqrt{Q} d^nt}{2\sqrt{K}\|\vrho\|_F} +C \ol D_{\text{PEPO}}\log(1+qp)}, & t > \frac{c_4\sqrt{QK} }{d^n}\|\vrho\|_F
  \end{cases}\nonumber\\
    &\!\!\!\!\leq\!\!\!\!& e^{-\min\{\frac{c_2 d^{2n}t^2}{4K\|\vrho\|_F^2},\frac{c_3 \sqrt{Q} d^nt}{2\sqrt{K}\|\vrho\|_F} \} + C \ol D_{\text{PEPO}}\log(1+qp) },
\end{eqnarray}
where $C$ and $c_i, i=2,3,4$ are positive constants and $\ol D_{\text{PEPO}} = \sum_{a=1}^{q}\sum_{b = 1}^{p}d^2 r_{\ul{a}\,\ul{b-1},\,\ul{a}\,\ul{b}} \;
r_{\ul{a-1}\,\ul{b},\,\ul{a}\,\ul{b}} \;
r_{\ul{a}\,\ul{b},\,\ul{a}\,\ul{b+1}} \;
r_{\ul{a}\,\ul{b},\,\ul{a+1}\,\ul{b}}$. When $Q =  \Omega(D_{\text{PEPO}})$ with $D_{\text{PEPO}} = \ol D_{\text{PEPO}} \log(1+qp) $, we have
\begin{eqnarray}
    \label{ProofOfSmallBallMethod_For_Haar_Measurement_W_2}
    \P{ T \geq t} \leq e^{-c_5 \frac{d^n t \sqrt{D_{\text{PEPO}}}}{\sqrt{K}\|\vrho\|_F}}, \ \forall \ t \geq c_6\frac{\sqrt{KD_{\text{PEPO}}}}{d^n}\|\vrho\|_F,
\end{eqnarray}
where $c_5$ and $c_6$ are positive constants. This further implies
\begin{align}
\label{ProofOfSmallBallMethod_For_Haar_Measurement_W_3}
 W(\ddot\setX_{\mR}) &= \E T \nonumber\\
    &\leq c_6 \frac{\sqrt{KD_{\text{PEPO}}}}{d^n}\|\vrho\|_F +  \int_{c_6 \frac{\sqrt{KD_{\text{PEPO}}}}{d^n}\|\vrho\|_F}^{\infty}\P{T\geq t} dt\nonumber\\
    &\le c_6 \frac{\sqrt{KD_{\text{PEPO}}}}{d^n}\|\vrho\|_F +  \int_{c_6 \frac{\sqrt{KD_{\text{PEPO}}}}{d^n}\|\vrho\|_F}^{\infty} e^{-c_5 \frac{d^n t \sqrt{D_{\text{PEPO}}}}{\sqrt{K}\|\vrho\|_F}} dt\nonumber\\
    &\leq  c_7 \frac{\sqrt{KD_{\text{PEPO}}}}{d^n}\|\vrho\|_F,
\end{align}
where $c_7$ is a positive constant.

Combining \eqref{ProofOfSmallBallMethod_For_Haar_Measurement PEPO_H upper bound}  and \eqref{ProofOfSmallBallMethod_For_Haar_Measurement_W_3}, and setting $t=\frac{c_0\sqrt{Q}}{2}$, $\xi =\frac{c_1}{d^n}\|\vrho\|_F$,   and $Q\geq \frac{64c_7^2 D_{\text{PEPO}}}{c_0^2c_1^2}$, we get
\begin{eqnarray}
    \label{ProofOfSmallBallMethod_For_Haar_Measurement_final}
    &\!\!\!\!\!\!\!\!&\inf_{\vrho\in\ddot\setX_{\mR}}\bigg(\sum_{i=1}^Q\sum_{k=1}^K|\< \vphi_{i,k}\vphi_{i,k}^\dagger, \vrho \>  |^2   \bigg)^{\frac{1}{2}} \nonumber\\
    &\!\!\!\!\geq\!\!\!\!& \xi \sqrt{QK} H_{\xi}(\ddot\setX_{\mR}) -2W(\ddot\setX_{\mR}) -t\xi\sqrt{K}\nonumber\\
    &\!\!\!\!\geq\!\!\!\!&\frac{c_0c_1\sqrt{QK}}{d^n}\|\vrho\|_F-2c_7 \frac{\sqrt{KD_{\text{PEPO}}} }{d^n}\|\vrho\|_F-\frac{c_1 \sqrt{QK}}{d^n}\|\vrho\|_F\nonumber\\
    &\!\!\!\!\geq\!\!\!\!&\frac{c_0c_1\sqrt{QK}}{4\cdot d^{n}}\|\vrho\|_F
\end{eqnarray}
with probability $1-e^{-\Omega(Q)}$.

For the PEPS set  $\wh\setX_{\mT}$, we can further construct a subset $\wt\setX_{\mT} = \{\vrho^{(h)}\in\C^{d^{n}\times d^{n}}: \vrho^{(h)} = \vu^{(h)} {\vu^{(h)}}^\dagger, \|\vu^{(h)}\|_2=1,  \vu^{(h)} = [\mU_{\ul{a}\, \ul{b}}^{(h_{\ul{a}\,\ul{b}})}]_{a=1,b=1}^{q,p}, \text{rank}(\vu^{(h)}) = \mT,\|\mU_{\ul{a}\, \ul{b}}^{(h_{\ul{a}\,\ul{b}})}\|_F \leq H_{ab}, 1\leq h_{\ul{a}\,\ul{b}}\leq |\wt\setX_{\mT}|, a\in[q], b\in[p] \}\subset \ddot\setX_{\mT} = \{\vrho\in\C^{d^{n}\times d^{n}}: \vrho = \vu \vu^\dagger, \vu = [\mU_{\ul{a}\, \ul{b}}]_{a=1,b=1}^{q,p},  \text{rank}(\vu) = \mT,\|\mU_{\ul{a}\, \ul{b}}\|_F \leq H_{ab}, a\in[q], b\in[p]  \}$ such that
\begin{eqnarray}
    \label{covering number for the PEPS subset}
    |\widetilde\setX_{\mT} | \leq \Pi_{a = 1}^{q}\Pi_{b = 1}^{p} \bigg(\frac{2+\epsilon_{ab}}{\epsilon_{ab}}\bigg)^{ d t_{\ul{a}\,\ul{b-1},\,\ul{a}\,\ul{b}}\;
    t_{\ul{a-1}\,\ul{b},\,\ul{a}\,\ul{b}}\;
    t_{\ul{a}\,\ul{b},\,\ul{a}\,\ul{b+1}}\;
    t_{\ul{a}\,\ul{b},\,\ul{a+1}\,\ul{b}}},
\end{eqnarray}
where $\epsilon_{ab}$ is defined by  $    \sup_{\mU_{\ul{a}\,\ul{b}}: \|\mU_{\ul{a}\,\ul{b}}\|_F\leq H_{ab}}\min_{h_{\ul{a}\,\ul{b}}\leq |\wt\setX_{\mT}|} \|\mU_{\ul{a}\,\ul{b}}-\mU_{\ul{a}\,\ul{b}}^{(h_{\ul{a}\,\ul{b}})}\|_F\leq H_{ab}\epsilon_{ab}$.

Using the decomposition $\vrho - \vrho^{(h)} = \vu\vu^\dagger - {\vu^{(h)}}{\vu^{(h)}}^\dagger = (\vu - {\vu^{(h)}})\vu^\dagger + {\vu^{(h)}}(\vu - {\vu^{(h)}})^\dagger$ and applying the same expansion as in \eqref{ProofOf<H,X>PEPO upper}, we can follow an analogous line of analysis using the covering argument. Then, with probability at least $1-e^{-\Omega(Q)}$, we obtain
\begin{eqnarray}
    \label{Small_Method_ROHaar_Measurement_Conclusion_Theorem PEPO appendix}
    \|\calA(\vrho)\|_2^2  = \sum_{i=1}^Q\sum_{k=1}^{d^n} |\<\vphi_{i,k} \vphi_{i,k}^\dagger, \vrho  \>|^2   \geq \Omega\bigg(\frac{Q}{d^{n}}\|\vrho\|_F^2\bigg)
\end{eqnarray}
provided that
$Q \geq \Omega(D_{\text{PEPS}})$, where $D_{\text{PEPS}} = \sum_{a=1}^{q}\sum_{b = 1}^{p} d \, t_{\ul{a}\,\ul{b-1},\,\ul{a}\,\ul{b}} \, t_{\ul{a-1}\,\ul{b},\,\ul{a}\,\ul{b}} \, t_{\ul{a}\,\ul{b},\,\ul{a}\,\ul{b+1}} \, t_{\ul{a}\,\ul{b},\,\ul{a+1}\,\ul{b}} \, \log(1+qp)$.
This completes the proof.

\end{proof}

\section{Proof of \Cref{final conclusion of 3-designs theorem PEPS,final conclusion of 3-designs theorem}}
\label{Proof of 2designs M}

\begin{proof}
In the proof, without loss of generality, we assume that
\begin{eqnarray}
    \label{The set of the PEPS in proof}
\text{PEPSs: }  \wh{\vrho}, \vrho^\star \in   \setX_{\mT} &\!\!\!\!=\!\!\!\!& \{\vrho\in\C^{d^{n}\times d^{n}}: \vrho = \vu \vu^\dagger, \|\vu\|_2=1, n = qp, \vu = [\mU_{\ul{a}\, \ul{b}}]_{a=1,b=1}^{q,p},\nonumber\\
     &\!\!\!\!\!\!\!\!& \text{rank}(\vu) = \mT,\|\mU_{\ul{a}\, \ul{b}}\|_F \leq H_{ab}, a\in[q], b\in[p] \},\\
    \label{The set of the PEPO in proof}
 \text{PEPOs: }  \wh{\vrho}, \vrho^\star \in   \setX_{\mR} &\!\!\!\!=\!\!\!\!& \{\vrho\in\C^{d^{n}\times d^{n}}: \vrho = \vrho^\dagger, \trace(\vrho) = 1, n = qp, \vrho = [\mX_{\ul{a}\, \ul{b}}]_{a=1,b=1}^{q,p},\nonumber\\
      &\!\!\!\!\!\!\!\!&\text{rank}(\vrho) = \mR,\|\mX_{\ul{a}\, \ul{b}}\|_F \leq L_{ab}, a\in[q], b\in[p]  \}.
\end{eqnarray}
While bounded factors are introduced solely to simplify the subsequent analysis, we note that in practice $\{ H_{ab} \}$ and $\{ L_{ab} \}$ can be arbitrary. Moreover, since a PEPS can be regarded as a special case of a PEPO, we first establish the result in the general PEPO setting and then restrict it to PEPS.

Now, we start by upper-bounding $\<  \veta, \calA(\wh{\vrho} - \vrho^\star) \>$. Towards that goal, we first rewrite this term as
\begin{eqnarray}
    \label{upper bound of entire variable_SIC_POVM}
    \<  \veta, \calA(\wh{\vrho} - \vrho^\star) \> = \sum_{k=1}^{K} \eta_{k}\<\mA_{k}, (\wh{\vrho} - \vrho^\star) \> \leq  \max_{\vrho\in \check\setX_{2\mR} } \sum_{k=1}^{K} \eta_{k}\<\mA_{k}, \vrho \>,
\end{eqnarray}
with
\begin{eqnarray}
    \label{The set of the PEPO normalized another}
    \check\setX_{2\mR} &\!\!\!\!=\!\!\!\!& \{\vrho\in\C^{d^n\times d^n}: \vrho = \vrho^\dagger, \trace(\vrho) = 0,  \vrho = [\mX_{\ul{a} \, \ul{b}}]_{a=1,b=1}^{q,p}, \nonumber\\
     &\!\!\!\!\!\!\!\!&\text{rank}(\vrho) = 2\mR, \|\mX_{\ul{a} \, \ul{b}}\|_F \leq 2L_{ab}, a\in[q], b\in[p]  \}.
\end{eqnarray}

The rest of the proof is to bound $\max_{\vrho\in \check\setX_{2\mR} } \sum_{k=1}^{K} \eta_{k}\<\mA_{k}, \vrho \>$, which will be achieved by using a covering argument.
First, when conditioned on $\{\mA_{k},\forall k \}$, we consider any fixed value of $\widetilde\vrho\in \{\wt{\check\setX}_{2\mR}\} \cap\{\widetilde\vrho: \|\widetilde\vrho\|_F\leq  \|\wh{\vrho} - \vrho^\star\|_F \} \subset \check\setX_{2\mR}$ and  apply \Cref{General bound of multinomial distribution Q cases} to establish a concentration inequality for the expression $\sum_{k=1}^{K} \eta_{k}\<\mA_{k}, \widetilde\vrho \>$. Specifically, we have
\begin{eqnarray}
    \label{Concentration inequality of sum of multinomial for t-designs}
    \P{\sum_{k=1}^{K} \eta_{k}\<\mA_{k}, \widetilde\vrho \> \geq t}&\!\!\!\!\leq\!\!\!\!&  e^{-\frac{Mt}{4\max_{k}|\<\mA_{k}, \widetilde\vrho \>|}\min\bigg\{1, \frac{\max_{k}|\<\mA_{k}, \widetilde\vrho \>| t}{4\sum_{k = 1}^{K}\<\mA_{k}, \widetilde\vrho \>^2p_{k} } \bigg\}  } +  e^{-\frac{Mt^2}{8\sum_{k = 1}^{K}\<\mA_{k}, \widetilde\vrho \>^2p_{k} }}\nonumber\\
    &\!\!\!\!\leq\!\!\!\!&  e^{-\frac{Mt^2}{16\sum_{k = 1}^{K}\<\mA_{k}, \widetilde\vrho \>^2p_{k} }} +  e^{-\frac{Mt^2}{8\sum_{k = 1}^{K}\<\mA_{k}, \widetilde\vrho \>^2p_{k} }},
\end{eqnarray}
where without loss of generality, we assume that $\frac{\max_{k}|\<\mA_{k}, \widetilde\vrho \>| t}{4\sum_{k = 1}^{K}\<\mA_{k}, \widetilde\vrho \>^2p_{k} }\leq 1$ in the last line.

Based on \Cref{property of 3-designs}, for $\wt\vrho\in\wt{\check\setX}_{2\mR} \subset \check\setX_{2\mR}$ in \eqref{The set of the PEPO normalized another}, we can obtain
\begin{eqnarray}
    \label{3 terms in the 3 designs conclusion}
    \sum_{k = 1}^{K}\<\mA_k, \wt\vrho \>^2p_{k} = O\bigg(\frac{\|\wh{\vrho} - \vrho^\star\|_F^2}{K^2}\bigg).
\end{eqnarray}
So we can rewrite \eqref{Concentration inequality of sum of multinomial for t-designs} as follows:
\begin{eqnarray}
    \label{Concentration inequality of sum of multinomial for 2-designs: final}
    \P{\sum_{k=1}^{K} \eta_{k}\<\mA_{k}, \widetilde\vrho \> \geq t}\leq  2e^{-\frac{K^2(d^n+1)Mt^2}{16 \cdot d^n \|\wh{\vrho} - \vrho^\star\|_F^2}}.
\end{eqnarray}

Following the same derivation of \eqref{ProofOfSmallBallMethod_For_Haar_Measurement_W_1}, there exists an $\epsilon$-net $\wt{\check\setX}_{2\mR}$ of $\check\setX_{2\mR}$ in \eqref{The set of the PEPO normalized another}  such that
\begin{eqnarray}
    \label{Upper bound of 2-deisgns covering arument1}
    &\!\!\!\!\!\!\!\!&\P{\max_{\vrho\in \check\setX_{2\mR}} \sum_{k=1}^{K} \eta_{k}\<\mA_{k}, \vrho \> \geq t  }\nonumber\\
    &\!\!\!\!\!\!\!\!&\hspace{-0.4cm} \leq\P{\max_{\widetilde\vrho\in \wt{\check\setX}_{2\mR}} \sum_{k=1}^{K} \eta_{k}\<\mA_{k}, \widetilde\vrho \> \geq \frac{t}{2}}\nonumber\\
    &\!\!\!\!\!\!\!\!&\hspace{-0.4cm}\leq\bigg(4qp+1\bigg)^{\ol D_{\text{PEPO}}  } e^{-\frac{K^2(d^n+1)Mt^2}{64 \cdot d^n \|\wh{\vrho} - \vrho^\star\|_F^2} + \log 2}\nonumber\\
    &\!\!\!\!\!\!\!\!&\hspace{-0.4cm}\leq e^{-\frac{K^2(d^n+1)Mt^2}{64 \cdot d^n \|\wh{\vrho} - \vrho^\star\|_F^2} +C D_{\text{PEPO}} + \log 2},
\end{eqnarray}
where $\ol D_{\text{PEPO}} = \sum_{a=1}^{q}\sum_{b = 1}^{p}d^2 r_{\ul{a}\,\ul{b-1},\,\ul{a}\,\ul{b}} \;
r_{\ul{a-1}\,\ul{b},\,\ul{a}\,\ul{b}} \;
r_{\ul{a}\,\ul{b},\,\ul{a}\,\ul{b+1}} \;
r_{\ul{a}\,\ul{b},\,\ul{a+1}\,\ul{b}}$ in the second inequality. In addition, $D_{\text{PEPO}} = \ol D_{\text{PEPO}} \log(1+qp)$ and $C$ is a universal constant in the last line. By taking $t = \frac{c_1  \sqrt{d^n  D_{\text{PEPO}}} }{K\sqrt{(d^n +1) M}}\|\wh{\vrho} - \vrho^\star\|_F$, we can further derive
\begin{eqnarray}
    \label{Concentration inequality of sum of multinomial for 2-designs: conclusion1}
    \P{\sum_{k=1}^{K} \eta_{k}\<\mA_{k}, \vrho \> \leq  \frac{c_1  \sqrt{d^n  D_{\text{PEPO}}} \|\wh{\vrho} - \vrho^\star\|_F}{K\sqrt{(d^n +1) M}}}\geq 1 - e^{- c_2 D_{\text{PEPO}}},
\end{eqnarray}
where $c_1$ and $c_2$ are constants.

Combing \eqref{whrho and rho^star relationship_1}, \eqref{The l2 norm of A(rho) 2_designs} and \eqref{Concentration inequality of sum of multinomial for 2-designs: conclusion1}, we have
\begin{eqnarray}
    \label{Concentration inequality of sum of multinomial for 2-designs: conclusion2}
    \|\wh{\vrho} - \vrho^\star\|_F \leq \frac{c_3  \sqrt{  D_{\text{PEPO}}} }{\sqrt{ M}},
\end{eqnarray}
where $c_3$ is a constant.

Combining the discussion in the final part of \Cref{Proof of Small ball PEPO} with the preceding analysis with the previous analysis, we can directly derive the corresponding result for PEPS as follows:
\begin{eqnarray}
    \label{Concentration inequality of sum of multinomial for 2-designs: conclusion2 PEPS}
    \|\wh{\vrho} - \vrho^\star\|_1\leq 2\|\wh{\vrho} - \vrho^\star\|_F \leq \frac{c_4  \sqrt{  D_{\text{PEPS}}} }{\sqrt{ M}},
\end{eqnarray}
where $c_4$ is a constant and $D_{\text{PEPS}} = \sum_{a=1}^{q}\sum_{b = 1}^{p} d \,
t_{\ul{a}\,\ul{b-1},\,\ul{a}\,\ul{b}} \,
t_{\ul{a-1}\,\ul{b},\,\ul{a}\,\ul{b}} \,
t_{\ul{a}\,\ul{b},\,\ul{a}\,\ul{b+1}} \,
t_{\ul{a}\,\ul{b},\,\ul{a+1}\,\ul{b}}
\log(1 + qp)$.

\end{proof}

\section{Proof of \Cref{final conclusion of Haar theorem PEPS,final conclusion of Haar theorem}}
\label{Proof of final conclusion of Haar theorem}

\begin{proof}
In the proof, without loss of generality, we assume that $\wh{\vrho}, \vrho^\star$ belong to $\setX_{\mT}$ for PEPSs \eqref{The set of the PEPS in proof} and to $\setX_{\mR}$ for PEPOs \eqref{The set of the PEPO in proof}. Noting that PEPS can be regarded as a special case of PEPO, we first establish the result in the general PEPO setting and then specialize it to PEPS.

According to \Cref{Small_Method_ROHaar_Measurement PEPO}, given $Q\geq \Omega( \sum_{a=1}^{q}\sum_{b = 1}^{p} d^2 r_{\ul{a}\,\ul{b-1},\ul{a}\,\ul{b}} \, r_{\ul{a-1}\,\ul{b},\ul{a}\,\ul{b}} \, r_{\ul{a}\,\ul{b},\ul{a}\,\ul{b+1}} \, r_{\ul{a}\,\ul{b},\ul{a+1}\,\ul{b}} \log(1+qp) )$, with probability at least $1-e^{-c_1Q}$, we have $\|\calA(\wh{\vrho} - \vrho^\star)\|_2^2\geq \Omega( \frac{Q}{d^{n}}\|\wh{\vrho} - \vrho^\star\|_F^2)$. Next, we will upper bound $\<  \veta, \calA(\wh{\vrho} - \vrho^\star) \>$. Towards that goal, we first rewrite this term as
\begin{eqnarray}
    \label{upper bound of entire variable_haar}
    \<  \veta, \calA(\wh{\vrho} - \vrho^\star) \> &\!\!\!\!\!\!=\!\!\!\!\!\!& \sum_{i=1}^Q\sum_{k=1}^{d^n} \eta_{i,k}\vphi_{i,k}^\dagger (\wh{\vrho} - \vrho^\star) \vphi_{i,k}\nonumber\\
    &\!\!\!\!\!\!\leq\!\!\!\!\!\!&   \max_{\vrho\in \check\setX_{2\mR}} \sum_{i=1}^Q\sum_{k=1}^{d^n} \eta_{i,k}\vphi_{i,k}^\dagger \vrho \vphi_{i,k},
\end{eqnarray}
where $\check\setX_{2\mR}$ is defined in \eqref{The set of the PEPO normalized another}.

The rest of the proof is to bound $\max_{\vrho\in \check\setX_{2\mR}} \sum_{i=1}^Q\sum_{k=1}^{d^n} \eta_{i,k}\vphi_{i,k}^\dagger \vrho \vphi_{i,k}$, which will be achieved by using a covering argument.
First, when conditioned on $\{\vphi_{i,k},\forall i,k \}$, we consider any fixed value of $\widetilde\vrho\in \{\wt{\check\setX}_{2\mR}\} \cap\{\widetilde\vrho: \|\widetilde\vrho\|_F\leq  \|\wh{\vrho} - \vrho^\star\|_F \} \subset \check\setX_{2\mR}$ and  apply \Cref{General bound of multinomial distribution Q cases} to establish a concentration inequality for the expression $\sum_{i=1}^Q\sum_{k=1}^{d^n} \eta_{i,k}\vphi_{i,k}^\dagger \widetilde\vrho \vphi_{i,k}$.
Following the proof in \cite[Appendix E-A]{qin2024quantum}, we denote the event $F:=\{ \max_{i,k} |\vphi_{i,k}^\dagger\wt\vrho\vphi_{i,k}| \lesssim \frac{\log Q + n\log d}{d^{\nqbit }}\|\wh{\vrho} - \vrho^\star\|_F$ $\}$ which holds with probability $\P{F} = 1-e^{-c_2 (\log Q + n\log d) }$ and then have
\begin{align}
    \label{tail function of of entire variable_haar original}
    \P{ \sum_{i=1}^Q\sum_{k=1}^{d^n} \eta_{i,k}\vphi_{i,k}^\dagger \widetilde\vrho \vphi_{i,k} \geq t \bigg| F  }\leq 2e^{-\frac{d^{2n}Mt^2}{c_3Q(\log Q + n\log d)^2\|\wh{\vrho} - \vrho^\star\|_F^2}},
\end{align}
where $c_2$ and $c_3$ are positive constants.

Following the same analysis as \eqref{ProofOfSmallBallMethod_For_Haar_Measurement_W_1}, there exists an $\epsilon$-net $\wt{\check\setX}_{2\mR}$ of $\check\setX_{2\mR}$  such that
\begin{eqnarray}
    \label{Upper bound of entire variable_haar covering arument1}
    &\!\!\!\!\!\!\!\!&\P{\max_{\vrho\in \check\setX_{2\mR}} \sum_{i=1}^Q\sum_{k=1}^{d^n} \eta_{i,k}\vphi_{i,k}^\dagger \vrho \vphi_{i,k}\geq t \bigg| F  }\nonumber\\
    &\!\!\!\!\!\!\!\!&\hspace{-0.4cm} \leq\P{\max_{\widetilde\vrho\in \wt{\check\setX}_{2\mR}} \sum_{i=1}^Q\sum_{k=1}^{d^n} \eta_{i,k}\vphi_{i,k}^\dagger \widetilde\vrho \vphi_{i,k}\geq \frac{t}{2} \bigg| F}\nonumber\\
    &\!\!\!\!\!\!\!\!&\hspace{-0.4cm}\leq\bigg(4qp+1\bigg)^{\ol D_{\text{PEPO}}  } e^{-\frac{d^{2n}Mt^2}{c_3Q(\log Q + n\log d)^2\|\wh{\vrho} - \vrho^\star\|_F^2} + \log 2}\nonumber\\
    &\!\!\!\!\!\!\!\!&\hspace{-0.4cm}\leq e^{-\frac{d^{2n}Mt^2}{c_3Q(\log Q + n\log d)^2\|\wh{\vrho} - \vrho^\star\|_F^2} +C D_{\text{PEPO}} + \log 2},
\end{eqnarray}
where $\ol D_{\text{PEPO}}  = \sum_{a=1}^{q}\sum_{b = 1}^{p}d^2 r_{\ul{a}\,\ul{b-1},\ul{a}\,\ul{b}} \, r_{\ul{a-1}\,\ul{b},\ul{a}\,\ul{b}} \, r_{\ul{a}\,\ul{b},\ul{a}\,\ul{b+1}} \, r_{\ul{a}\,\ul{b},\ul{a+1}\,\ul{b}}$ in the second inequality. In addition, $D_{\text{PEPO}} = \ol D_{\text{PEPO}} \log(1+qp)$ and $C$ is a universal constant in the last line. By taking $t=\hat{t} \triangleq \frac{c_4  \sqrt{QD_{\text{PEPO}}}(\log Q + n \log d)}{\sqrt{M}d^{n}}\|\wh{\vrho} - \vrho^\star\|_F$ in the above equation, we further obtain
\begin{align}
    \label{Upper bound of entire variable_haar covering arument another1}
    \hspace{-0.1cm}\P{\max_{\vrho\in \check\setX_{2\mR}} \sum_{i=1}^Q\sum_{k=1}^{d^n} \eta_{i,k}\vphi_{i,k}^\dagger \vrho \vphi_{i,k}\leq \hat{t} \bigg| F  }\geq 1-e^{-c_5 D_{\text{PEPO}}},
\end{align}
where $c_4$ and $c_5$ are  constants.

Now plugging in the probability for the event $F$, we finally get
\begin{eqnarray}
    \label{Upper bound of entire variable_haar covering arument another 11}
    &\!\!\!\!\!\!\!\!&\P{\max_{\vrho\in \check\setX_{2\mR}} \sum_{i=1}^Q\sum_{k=1}^{d^n} \eta_{i,k}\vphi_{i,k}^\dagger \vrho \vphi_{i,k}\leq \hat{t}}\nonumber\\
    &\!\!\!\!\geq\!\!\!\!& \P{\max_{\vrho\in \check\setX_{2\mR}} \sum_{i=1}^Q\sum_{k=1}^{d^n} \eta_{i,k}\vphi_{i,k}^\dagger \vrho \vphi_{i,k}\leq \hat{t} \cap F  }\nonumber\\
    &\!\!\!\!=\!\!\!\!&\P{F} \P{\max_{\vrho\in \check\setX_{2\mR}} \sum_{i=1}^Q\sum_{k=1}^{d^n} \eta_{i,k}\vphi_{i,k}^\dagger \vrho \vphi_{i,k}\leq  \hat{t} \bigg| F  }\nonumber\\
    &\!\!\!\!\geq\!\!\!\!& (1-e^{-c_2 \log(Qd^n) })(1-e^{-c_5 D_{\text{PEPO}}})\nonumber\\
    &\!\!\!\!\geq\!\!\!\!& 1- e^{-c_2 (\log Q + n\log d) } - e^{-c_5 D_{\text{PEPO}}}.
\end{eqnarray}

Hence, for $\<  \veta, \calA(\wh{\vrho} - \vrho^\star) \>$ in \eqref{upper bound of entire variable_haar}, the above equation implies that with probability at least $1- e^{-c_2 (\log Q + n\log d) } - e^{-c_5 D_{\text{PEPO}}}$,
\begin{eqnarray}
    \label{upper bound of entire variable_haar conclusion1}
    \<  \veta, \calA(\wh{\vrho} - \vrho^\star) \> \leq  \frac{c_4  \sqrt{QD_{\text{PEPO}}}(\log Q + n \log d)}{\sqrt{M}d^{n}}\|\wh{\vrho} - \vrho^\star\|_F.
\end{eqnarray}
Combining this together with $\|\calA(\wh{\vrho} - \vrho^\star)\|_2^2 \geq \Omega( \frac{Q}{d^{n}}\|\wh{\vrho} - \vrho^\star\|_F^2)$, we finally obtain
\begin{eqnarray}
    \label{statistical error bound conclusion1}
    \|\wh{\vrho} - \vrho^\star\|_F\leq O\bigg( \frac{ \sqrt{D_{\text{PEPO}}}(\log Q + n \log d)}{\sqrt{MQ}} \bigg).
\end{eqnarray}

Combining the discussion in the final part of \Cref{Proof of Small ball PEPO} with the preceding analysis with the previous analysis, we immediately obtain the corresponding statistical error bound for PEPS:
\begin{eqnarray}
    \label{statistical error bound conclusion1 PEPS}
    \|\wh{\vrho} - \vrho^\star\|_1\leq 2\|\wh{\vrho} - \vrho^\star\|_F\leq O\bigg( \frac{ \sqrt{D_{\text{PEPS}}}(\log Q + n \log d)}{\sqrt{MQ}} \bigg),
\end{eqnarray}
where  $D_{\text{PEPS}} = \sum_{a=1}^{q}\sum_{b = 1}^{p} d \,
t_{\ul{a}\,\ul{b-1},\,\ul{a}\,\ul{b}} \,
t_{\ul{a-1}\,\ul{b},\,\ul{a}\,\ul{b}} \,
t_{\ul{a}\,\ul{b},\,\ul{a}\,\ul{b+1}} \,
t_{\ul{a}\,\ul{b},\,\ul{a+1}\,\ul{b}}
\log(1 + qp)$.

\end{proof}

\section{Auxiliary Materials}
\label{Auxiliary materials}

\begin{lemma}(\cite[Lemma 2]{qin2024quantum})
\label{Small Ball Method_abs_New}
Consider a fixed set $E\subset \C^{D}$. Let $\{\vb_{1},\!\dots,\!\vb_K\!\}$ represent a collection of random columns in $\C^{D}$, which may not be mutually independent. Additionally, let $\{\vb_{i,1},\dots,\vb_{i,K}\}_{i=1}^Q$ denote a set of independent copies of $\{\vb_{1},\dots,\vb_K\}$.
Introduce the marginal tail function
\begin{eqnarray}
    \label{marginal_tail_function_New}
    H_{\xi}(E;\vb) = \inf_{\vu\in E} \frac{1}{K}\sum_{k=1}^K\mathbb{P} \{|\<\vb_{k},\vu \>|\geq \xi  \},\ for \ \xi>0.
\end{eqnarray}
Let $\epsilon_i,i=1,\dots,Q$ be independent Rademacher random variables, independent from everything else, and define the
mean empirical width of the set:
\begin{align}
    \label{mean empirical width_New}
    W_{QK}(E;\vb) = \E \sup_{\vu\in E}\<\vh ,\vu\>, \text{where} \ \vh =\frac{1}{\sqrt{QK}}\sum_{i=1}^Q\sum_{k=1}^K \epsilon_{i}\vb_{i,k}.
\end{align}

Then, for any $\xi >0$ and $t >0$
\begin{eqnarray}
    \label{Final_Conclusion_New}
    \hspace{-0.5cm}\inf_{\vu\in E}\bigg(\sum_{i=1}^Q\sum_{k=1}^K|\<\vb_{i,k}, \vu  \>|^2  \bigg)^{\frac{1}{2}} \geq \xi\sqrt{QK}H_{\xi}(E;\vb)   -2W_{QK}(E;\vb) -t \xi\sqrt{K},
\end{eqnarray}
with probability at least $1-e^{-\frac{t^2}{2}}$.

\end{lemma}

\begin{lemma} (\cite[Lemma 3]{dall2014accessible})
\label{property of 2-designs}
The integral over the Haar measure in \eqref{the definition of t designs} is given by
\begin{eqnarray} \label{the definition of integral in t designs}
     \int(\vw\vw^\dagger)^{\otimes s} d\vw = \frac{1}{C_{d^n+s-1}^s} P_{\text{Sym}},
\end{eqnarray}
where $P_{\text{Sym}}$ is the projector over the symmetric subspace.
\end{lemma}

\begin{lemma}
\cite[Theorem 3.1]{tanoue2022concentration}
\label{Concentration inequality of sum of dependent sub-exp sum}
Suppose that $X = \sum_{i=1}^Q \sum_{k=1}^K w_k X_{i,k}$, where $w_k,k=1,\dots, K$ are constants, and each $X_{i,k}, i=1,\dots,Q, k=1,\dots,K$ is a zero-mean, subexponential random variable with $\|X_{i,k} \|_{\psi_1}$. In addition, the $Q$ multivariate random variables $(X_{i,1},\ldots,X_{i,K}), i = 1,\ldots, Q$ are mutually independent. However, it is possible for the variables $X_{i,k}$ and $X_{i,k'}, k'\neq k$ within each multivariate random variable to be dependent.
Then
\begin{eqnarray}
    \label{Concentration inequality of sum of dependent sub-exp}
    \P{X>t}\leq \begin{cases}
    e^{-\frac{t^2}{4T^2}}, & t\leq 2T^2H,\\
    e^{-\frac{tH}{2}}, & t>2T^2H.
  \end{cases}
\end{eqnarray}
where $T = \sum_{k=1}^Kw_k \sqrt{\sum_{i=1}^Q c_{i,k} \|X_{i,k} \|_{\psi_1}^2 }$ and $H = \bigg(\min_{k}\frac{ \sqrt{\sum_{i=1}^Q c_{i,k} \|X_{i,k} \|_{\psi_1}^2 }}{ \sum_{k=1}^Kw_k \sqrt{\sum_{i=1}^Q c_{i,k} \|X_{i,k} \|_{\psi_1}^2 }}\bigg)\cdot\bigg(\min_{i} \{\frac{d_{i,k}}{\|X_{i,k} \|_{\psi_1}} \} \bigg)$ with constants $c_{i,k}$ and $d_{i,k}$.
\end{lemma}

\begin{lemma}(\cite[Lemma 14]{qin2024quantum})
\label{General bound of multinomial distribution Q cases}
Suppose that the $Q$ multivariate random variables $(f_{i,k},\dots, f_{i,K}),i=1,\dots,Q$ are mutually independent and follow the multinomial distribution $\operatorname{Multinomial}(M,\vp_i)$  with $\sum_{k=1}^{K}f_{i,k} =M $ and $\vp_i = [p_{i,1}.\dots, p_{i,K}]$, respectively.
Let $a_{i,1},\dots, a_{i,K}$ be fixed. Then, for any $t>0$,
\begin{align}
    \label{General bound of multinomial distribution for all constant Q cases}
    \P{\sum_{i=1}^Q\sum_{k=1}^Ka_{i,k}(\frac{f_{i,k}}{M} - p_{i,k}) > t   }\leq  e^{-\frac{Mt}{4a_{\max}}\! \min\bigg\{\! 1, \frac{a_{\max}t }{4\sum_{i=1}^Q\sum_{k=1}^K a_{i,k}^2p_{i,k}} \! \bigg\}}\! + \! e^{-\frac{Mt^2 }{8\sum_{i=1}^Q\sum_{k=1}^K a_{i,k}^2p_{i,k}}},
\end{align}
where $a_{\max} = \max_{i,k}|a_{i,k}|$.
\end{lemma}

\begin{lemma}
\label{property of 3-designs}
Suppose that the exact $3$-designs $\{\widetilde{p}_k,  \vw_k \}_{k=1}^K$ is uniformly distributed, namely, $\widetilde{p}_k = \frac{1}{K}, \forall k$. For any $\vrho, \vrho^\star\in\setX_{\mR}$ in \eqref{The set of the PEPO}, we have
\begin{eqnarray}
    \label{final upper bound 3 terms in the 3 designs}
    \sum_{k = 1}^{K}\<\frac{d^n}{K} \vw_k\vw_k^\dagger, \vrho - \vrho^\star  \>^2\<\frac{d^n}{K} \vw_k\vw_k^\dagger, \vrho^\star  \>= O\bigg(\frac{\|\vrho - \vrho^\star\|_F^2}{K^2}\bigg).
\end{eqnarray}
\end{lemma}
\begin{proof}
Leveraging the properties of $3$-designs POVM \cite{ambainis2007quantum,dall2014accessible}, we can rewrite $\sum_{k = 1}^{K}\<\frac{d^n}{K} \vw_k\vw_k^\dagger, \vrho - \vrho^\star  \>^2\<\frac{d^n}{K} \vw_k\vw_k^\dagger, \vrho^\star  \>$ as follows:
\begin{eqnarray}
    \label{3 terms in the 3 designs}
    &\!\!\!\!\!\!\!\!& \sum_{k = 1}^{K}\<\frac{d^n}{K} \vw_k\vw_k^\dagger, \vrho - \vrho^\star  \>^2\<\frac{d^n}{K} \vw_k\vw_k^\dagger, \vrho^\star  \>\nonumber\\
    &\!\!\!\!=\!\!\!\!& \sum_{k = 1}^{K} \frac{d^{3n}}{K^3} \trace((\vw_k\vw_k^\dagger)^{\otimes 3} (\vrho - \vrho^\star)\otimes (\vrho - \vrho^\star) \otimes \vrho^\star  )\nonumber\\
    &\!\!\!\!=\!\!\!\!&\frac{d^{3n}}{K^2}\frac{6}{(d^n+2)(d^n+1)d^n}\trace( \calP_{\text{sym}^{\otimes 3}} (\vrho - \vrho^\star)\otimes (\vrho - \vrho^\star) \otimes \vrho^\star),
\end{eqnarray}
where $\calP_{\text{sym}^{\otimes 3}}$, defined in \cite[Lemma 3]{dall2014accessible}, represents the projector onto the symmetric subspace.
Additionally, based on \cite[Lemma 7]{gross2015partial} and \cite[eq. (322)]{mele2023introduction}, we have
\begin{eqnarray}
    \label{projection of symmetric 3design}
    &\!\!\!\!\!\!\!\!&\trace( \calP_{\text{sym}^{\otimes 3}} (\vrho - \vrho^\star)\otimes (\vrho - \vrho^\star) \otimes \vrho^\star) \nonumber\\
    &\!\!\!\!=\!\!\!\!& \frac{1}{6}\bigg((\trace(\vrho - \vrho^\star ))^2\trace(\vrho^\star) + \trace((\vrho - \vrho^\star)^2 )\trace(\vrho^\star)\nonumber\\
    &\!\!\!\!\!\!\!\!& + 2\trace((\vrho - \vrho^\star)\vrho^\star )\trace(\vrho - \vrho^\star)  + 2\trace((\vrho - \vrho^\star)^2 \vrho^\star)\bigg)\nonumber\\
    &\!\!\!\!=\!\!\!\!& \frac{1}{6}\| \vrho - \vrho^\star\|_F^2 + \frac{1}{3}\trace((\vrho - \vrho^\star)^2 \vrho^\star) .
\end{eqnarray}
Consequently, we arrive at $\sum_{k = 1}^{K}\<\mA_k, \vrho - \vrho^\star \>^2p_{k} = O\big(\frac{\| \vrho - \vrho^\star\|_F^2 + \trace((\vrho - \vrho^\star)^2 \vrho^\star)}{K^2}\big)$. Given that $\vrho - \vrho^\star$ is Hermitian, we can ensure $(\vrho - \vrho^\star)^2 = (\vrho - \vrho^\star)(\vrho - \vrho^\star)^\dagger$ is PSD. Building upon \cite[Theorem 1]{coope1994matrix} using two PSD matrices $(\vrho - \vrho^\star)^2$ and $\vrho^\star$, we further deduce $\trace((\vrho - \vrho^\star)^2 \vrho^\star)\leq \trace((\vrho - \vrho^\star)^2)\trace(\vrho^\star) = \|\vrho - \vrho^\star\|_F^2$. Ultimately, we obtain
\begin{eqnarray}
    \label{3 terms in the 3 designs 1}
    \sum_{k = 1}^{K}\<\frac{d^n}{K} \vw_k\vw_k^\dagger, \vrho - \vrho^\star  \>^2\<\frac{d^n}{K} \vw_k\vw_k^\dagger, \vrho^\star  \> = O\bigg(\frac{\|\vrho - \vrho^\star\|_F^2}{K^2}\bigg).
\end{eqnarray}

\end{proof}

%{
%%\bibliographystyle{alpha}
%\bibliographystyle{unsrt}
%\bibliography{reference1}
%}

\end{document}

%% file: macro.tex
\usepackage{url}
\usepackage[hidelinks]{hyperref}
\usepackage{amsmath,amsthm,amssymb,amsbsy}
\usepackage{paralist}
\usepackage{xcolor}
\usepackage{color}
\usepackage{comment}
\usepackage{multirow}
\usepackage[toc, page]{appendix}

\usepackage{fancyhdr}
\usepackage{cite}
\usepackage{cleveref}
\newtheorem{lemma}{Lemma}%[section] %%    with section number.
\newtheorem{defi}{Definition}%[section]

\newtheorem{theorem}{Theorem}

\theoremstyle{remark}

\theoremstyle{problem}

\newcommand{\R}{\mathbb{R}}
\def \real    { \mathbb{R} }

\newcommand{\C}{\mathbb{C}}

\newcommand{\e}{\begin{equation}}
\newcommand{\ee}{\end{equation}}
\newcommand{\en}{\begin{equation*}}
\newcommand{\een}{\end{equation*}}
\newcommand{\eqn}{\begin{eqnarray}}
\newcommand{\eeqn}{\end{eqnarray}}
\newcommand{\bmat}{\begin{bmatrix}}
\newcommand{\emat}{\end{bmatrix}}
\DeclareMathAlphabet\mathbfcal{OMS}{cmsy}{b}{n}
\renewcommand{\P}[1]{\operatorname{\mathbb{P}}\left(#1\right)}
\newcommand{\E}{\operatorname{\mathbb{E}}}

\newcommand{\vct}[1]{\boldsymbol{#1}}
\newcommand{\mtx}[1]{\boldsymbol{#1}}
\newcommand{\<}{\langle}
\renewcommand{\>}{\rangle}
\newcommand{\trace}{\operatorname{trace}}

\newcommand{\set}[1]{\mathbb{#1}}
\DeclareMathOperator*{\argmin}{\text{arg~min}}

\def \st {\operatorname*{s.t.\ }}
\newcommand{\wh}{\widehat}

\newcommand{\wt}{\widetilde}
\newcommand{\ol}{\overline}
\newcommand{\ul}{\underline}
\newcommand{\nqbit}{n}

\newcommand{\innerprod}[2]{\left\langle #1,  #2 \right\rangle}

\newcommand{\calA}{\mathcal{A}}

\newcommand{\calP}{\mathcal{P}}

\newcommand{\va}{\vct{a}}
\newcommand{\vb}{\vct{b}}

\newcommand{\ve}{\vct{e}}

\newcommand{\vh}{\vct{h}}

\newcommand{\vp}{\vct{p}}

\newcommand{\vu}{\vct{u}}

\newcommand{\vw}{\vct{w}}

\newcommand{\vphi}{\vct{\phi}}

\newcommand{\veta}{\vct{\eta}}

\newcommand{\vrho}{\vct{\rho}}
\newcommand{\vzero}{\vct{0}}

\newcommand{\mA}{\mtx{A}}
\newcommand{\mB}{\mtx{B}}

\newcommand{\mR}{\mtx{R}}
\newcommand{\mS}{\mtx{S}}
\newcommand{\mT}{\mtx{T}}
\newcommand{\mU}{\mtx{U}}

\newcommand{\mX}{\mtx{X}}

\newcommand{\mId}{{\bf I}}

\newcommand{\setS}{\set{S}}

\newcommand{\setX}{\set{X}}

\graphicspath{{./figs/}}

\newlength{\imgwidth}
\newboolean{twoColVersion}
\setboolean{twoColVersion}{false}
\newcommand{\twoCol}[2]{\ifthenelse{\boolean{twoColVersion}} {#1} {#2} }